\documentclass[submission,copyright,creativecommons]{eptcs}
\providecommand{\event}{SOS 2007} 
\usepackage{breakurl}             
\usepackage{underscore}           

\usepackage{graphicx}
\usepackage{proof}
\usepackage{amsmath}
\usepackage{amsfonts}
\usepackage{bussproofs}
\usepackage{mathrsfs}
\usepackage{amssymb}
\usepackage{tikz}
\usepackage{enumerate}
\usepackage{xspace}
\usepackage{tipa}

\usepackage[utf8]{inputenc}

\newcommand{\logicP}[0]{\mathcal{P} }
\renewcommand{\logicP}[0]{\mathbf{P}}

\newcommand{\logicLP}[0]{\mathcal{LP} }
\renewcommand{\logicLP}[0]{\mathbf{LP} }

\newcommand{\logicFour}[0]{\mathcal{F} }
\renewcommand{\logicFour}[0]{\mathbf{F} }

\newcommand{\modelsPm}[0]{\models_{\logicP_m} }
\renewcommand{\modelsPm}[0]{\models_{\logicP}^\mathit{min} }
\newcommand{\modelsLPm}[0]{\models_{\logicLP_m} }
\renewcommand{\modelsLPm}[0]{\models_{\logicLP}^\mathit{min}}
\newcommand{\modelsFourMa}[0]{\models_{\logicFour_{m_1}} }
\renewcommand{\modelsFourMa}[0]{\models^{\mathit{min}_1}_{\logicFour} }
\newcommand{\modelsFourMb}[0]{\models_{\logicFour_{m_2}} }
\renewcommand{\modelsFourMb}[0]{\models^{\mathit{min}_2}_{\logicFour} }

\newcommand{\minModels}[0]{\models^\mathcal{I}_\mathcal{L} }
\newcommand{\minSequent}[0]{\Rightarrow^\mathcal{I}_\mathcal{L} }

\newcommand{\unaryOp}[0]{\mathbf{I}_\mathcal{L}^\mathcal{I} }
\newcommand{\unaryNotOp}[0]{\overline{\mathbf{I}_\mathcal{L}^\mathcal{I}} }
\renewcommand{\unaryNotOp}[0]{\mathbf{C}^\mathcal{I}_\mathcal{L} }
\newcommand{\unaryOpPm}[0]{\mathbf{I}_{\logicP}^{\{ \mathbf{b} \}} }
\newcommand{\unaryNotOpPm}[0]{\overline{\mathbf{I}_{\logicP}^{\{ \mathbf{b} \}}} }
\renewcommand{\unaryNotOpPm}[0]{\mathbf{C}_{\logicP}^{\{ \mathbf{b} \}} }
\newcommand{\unaryNotOpFa}[0]{\overline{\mathbf{I}_{\logicFour}^{\{ \mathbf{b} \}}} }
\renewcommand{\unaryNotOpFa}[0]{\mathbf{C}_{\logicFour}^{\{ \mathbf{b} \}} }
\newcommand{\unaryNotOpFb}[0]{\overline{\mathbf{I}_{\logicFour}^{\{ \mathbf{b}, \mathbf{n} \}}} }
\renewcommand{\unaryNotOpFb}[0]{\mathbf{C}_{\logicFour}^{\{ \mathbf{b}, \mathbf{n} \}} }

\newtheorem{definition}{Definition}[section]
\newtheorem{example}{Example}[section]
\newtheorem{theorem}{Theorem}[section]
\newtheorem{proposition}{Proposition}[section]

\newcommand{\Val}{\ensuremath{\mathcal{V}}\xspace}
\newcommand{\Log}{\ensuremath{\mathcal{L}}\xspace}
\newcommand{\val}{\ensuremath{\mathsf{v}}\xspace}
\newcommand{\conn}{\ensuremath{\circ}\xspace}
\newcommand{\const}{\ensuremath{\mathcal{P}}\xspace}
\newcommand{\defcon}{\ensuremath{\mathit{def}}\xspace}
\renewcommand{\defcon}{\ensuremath{\mathit{f}}\xspace}
\newcommand{\valu}{\ensuremath{\mathit{val}}\xspace}
\renewcommand{\valu}[2]{\ensuremath{v^{#1}_{#2}}\xspace}

\newcommand{\Lm}[0]{\Log_{m}}
\renewcommand{\Lm}[0]{\Log}
\newcommand{\Vm}[0]{\Val_{m}}
\renewcommand{\Vm}[0]{\Val_{\Lm}}

\newcommand{\commadots}[0]{,\ldots ,}
\newcommand{\egc}[0]{e.g.,\ }
\newcommand{\iec}[0]{i.e.,\ }

\newcommand{\CFalse}{\mathrm{F}}

\newcommand{\CBoth}{\mathrm{B}}
\newcommand{\CNeither}{\mathrm{N}}

\newcommand{\MEL}{\mathrm{ME}_\Lm^\mathcal{I}}

\def\halmos{\mbox{ }\hfill$\Box$}
\newenvironment{proof}[1]%
   {\par\noindent{\normalsize\emph{Proof}.} \ #1}%
   {\hfill\halmos
   \par}

\title{Sequent-Type Calculi for Systems of\\ Nonmonotonic Paraconsistent Logics}
\author{Tobias Geibinger
\institute{Databases and Artificial Intelligence Group,\\
Institute of Logic and Computation, \\
Technische Universit\"at Wien,\\
Favoritenstra\ss{}e\ 9-11,
A-1040 Vienna, Austria}
\email{tgeibing@dbai.tuwien.ac.at}
\and
Hans Tompits
\institute{Knowledge-Based Systems Group,\\
Institute of Logic and Computation, \\
Technische Universit\"at Wien,\\
Favoritenstra\ss{}e\ 9-11,
A-1040 Vienna, Austria}
\email{tompits@kr.tuwien.ac.at}
}
\def\titlerunning{Sequent-Type Calculi for Systems of Nonmonotonic Paraconsistent Logics}
\def\authorrunning{Tobias Geibinger \& Hans Tompits}
\begin{document}
\maketitle

\begin{abstract}
Paraconsistent logics constitute an important class of formalisms dealing with non-trivial reasoning from inconsistent premisses.
In this paper, we introduce uniform axiomatisations for a family of nonmonotonic paraconsistent logics based on minimal inconsistency 
in terms of sequent-type proof systems. 
The latter are prominent and widely-used forms of calculi 
well-suited for analysing proof search.
In particular, we provide sequent-type calculi for Priest's three-valued 
\emph{minimally inconsistent logic of paradox}, and for four-valued paraconsistent inference relations due to Arieli and Avron.
Our calculi follow the sequent method first introduced in the context of nonmonotonic reasoning by Bonatti and Olivetti, whose distinguishing feature is the use of a so-called \emph{rejection calculus} for axiomatising invalid formulas. In fact, we present
a general method to obtain sequent systems for any many-valued logic based on minimal inconsistency,
yielding the calculi for the logics of Priest and
of Arieli and Avron 
as special instances.
\end{abstract}

\section{Introduction}

Paraconsistent logics reject the principle of explosion, also known as \emph{ex falso sequitur quodlibet}, which holds in classical logic and allows the derivation of any assertion from a contradiction. The motivation behind paraconsistent logics is simple, as contradictory theories may still contain useful information, hence we would like to be able to draw non-trivial conclusions from said theories. This is of course also interesting in the context of artificial intelligence and especially in knowledge representation. Human knowledge is often contradictory and yet it allows us to reason about the world. 

The interest in nonmonotonic logics was born out of somewhat similar motivations---in particular, from the desire to formalise instances of common-sense reasoning which are difficult to express in classical logic without falling into contradiction and thus triviality. However nonmonotonic logics do not reject the principle of explosion but the monotony principle of classical logic. In those logics, inferences are in general defeasible, meaning that conclusions which have been previously drawn might not be derivable in the light of new information. 

In this paper, we introduce sequent-type proof systems for inference relations which are paraconsistent and nonmonotonic, based on propositional many-valued logics. 
The formalisms we consider are due to Priest~\cite{priest1991} and Arieli and Avron~\cite{arieli1998},
and their nonmonotonic flavour is obtained by a circumscription-like minimal-model reasoning, where models with less amount of inconsistency are in a sense preferred. 

In order to obtain calculi for the mentioned inference relations, we
adopt the sequent method of Bonatti and Olivetti~\cite{bonatti2002}, who introduced 
proof systems for the central nonmonotonic formalisms, viz.\ for default logic~\cite{reiter}, autoepistemic logic~\cite{moore-85}, and circumscription~\cite{McCarthy:1980}.
A key feature of their approach is their usage of a \emph{rejection calculus} for axiomatising invalid formulas, \iec of non-theorems, 
which makes these calculi arguably particularly elegant and suitable for proof-complexity elaborations~\cite{egly01,BeyersdorffMTV12}.
In a rejection calculus, the inference rules formalise the propagation
of refutability instead of validity and establish invalidity by deduction, \iec in a purely syntactic manner.
Rejection calculi are also referred to as \emph{complementary calculi} or \emph{refutation calculi} in the literature and the first axiomatic treatment of rejection was done by {\L}ukasiewicz~\cite{lukas39} in his formalisation of Aristotle's syllogistic.

Analogous to the method of Bonatti~\cite{bonatti2002}, our calculi
comprise three kinds of sequents each: 
(i)~assertional sequents for axiomatising validity in the respective underlying monotonic base logic,
(ii)~\emph{anti-sequents} for axiomatising {invalidity} for the underlying base logics,
and (iii)~sequents for representing nonmonotonic conclusions.

In fact, we prove a somewhat stronger result in that we not only provide calculi for said formalisms, but we give a uniform method to obtain such calculi for any many-valued entailment relation based on minimal inconsistency.

As far as calculi for many-valued logics are concerned, different kinds of sequent-style systems exist in the literature, like systems based on (two-sided) sequents~\cite{beziau99,Avron02} in the style of the original work by Gentzen~\cite{gentzen1935} and employing additional non-standard rules,
or using \emph{hypersequents}~\cite{avron1991}, 
which are tuples of Gentzen-style sequents.
In our sequent and anti-sequent calculi, we follow the approach of Rousseau~\cite{rousseau1967}, 
which is a natural generalisation for many-valued logics of the classical two-sided sequent formulation of Gentzen.
The respective calculi are obtained from a systematic construction for many-valued logics as described by Zach~\cite{zach1993} and by Bogojeski and Tompits~\cite{bogo-tompits20}.

It should be noted that other approaches exist for formalising the inference relations we study in this work. Arieli and Denecker~\cite{arieli2003} describe a method to encode a theory in Belnap's four-valued paraconsistent logic~\cite{belnap1977a} into a classical theory. They then use circumscription to model multiple minimal inconsistent inference relations. In a similar fashion, Besnard, Schaub, Tompits, and Woltran~\cite{besnardSTW05} encode theories in Priest's minimally inconsistent three-valued paraconsistent logic~\cite{priest1991} in terms of \emph{quantified boolean formulas} (QBF). In difference to those approaches, we do not rely on any encoding into another formalism but rather provide a direct proof-theoretic characterisation.

The rest of paper is organised as follows. In the next section, we establish the necessary preliminaries.
The general method to obtain sequent-type calculi for the inferences we are interested in is described in Section~\ref{sec:genMethod}. In Section~\ref{sec:calcLP}, 
we provide concrete sequent systems obtained through our approach. Finally, in Section~\ref{sec:concl}, we give some concluding remarks.

\section{Preliminaries}

\paragraph{Syntax and Semantics of Finite-Valued Propositional Logics.}

A \emph{finite-valued propositional logic}, $\Lm$, is defined over a set $\Vm=\{\val_1\commadots \val_m\}$
of \emph{truth values}, a set $\Vm^+\subset \Vm$ of \emph{designated truth values}  (which are used to define modelhood), and a vocabulary $\mathcal{A}_\Lm$ consisting of (i)~a countably infinite set $\const$ of \emph{propositional constants} and
(ii)~a collection of $n$-ary ($n\geq 0$) \emph{primitive logical connectives}.
We assume that $\Vm$ always contains the truth values $\mathbf{t}$ and $\mathbf{f}$ (representing \emph{truth} and \emph{falsity}, respectively) such that $\mathbf{t} \in \Vm^+$ and $\mathbf{f} \not\in \Vm^+$.
A 0-ary logical connective is called a \emph{logical constant}.
Furthermore, the set $\const$ is assumed fixed throughout this paper.

Formulas of the logic $\Lm$ are referred to as \emph{$\Lm$-formulas} and are inductively defined as follows: (i)~every propositional constant and every logical constant of $\mathcal{A}_\Lm$ is an $\Lm$-formula; (ii)~if $\varphi_1,\dots,\varphi_n$ are $\Lm$-formulas and $\conn$ is an $n$-ary connective of $\mathcal{A}_\Lm$ (for $n\geq 1$), then $\conn(\varphi_1,\dots,\varphi_n)$ is an $\Lm$-formula; and (iii)~$\Lm$-formulas are constructed only according to (i) and (ii).
In the following, binary connectives are usually written infix to increase readability.

 An \emph{$\Lm$-interpretation} is a mapping $I : \const \longrightarrow \Vm$ assigning to each propositional constant a truth value from $\Vm$.
For a set $\Theta\subseteq \const$, we write $I|_\Theta$ to denote the mapping resulting from $I$ by restricting the domain $\const$ to the propositional constants in $\Theta$.

Given an $\Lm$-interpretation $I$, by a \emph{valuation under $I$} we understand a mapping $\valu{I}{\Lm}(\cdot)$ which assigns to each $\Lm$-formula $\varphi$ a truth value of $\Vm=\{\val_1\commadots \val_m\}$ subject to the following conditions: (i)~if $\varphi$ is a propositional constant of $\Lm$, then $\valu{I}{\Lm}(\varphi) = I(\varphi)$; and 
(ii)~if $\varphi=\conn(\psi_1,\dots,\psi_n)$, for an $n$-ary logical connective ($n\geq 0$), then  $\valu{I}{\Lm}(\varphi)=\defcon_{\conn}
    (\valu{I}{\Lm}(\psi_1),\dots,\valu{I}{\Lm}(\psi_n))$, where $\defcon_{\conn}:\Vm^n\longrightarrow\Vm$ is a function representing the truth conditions of $\conn$ in $\Lm$ (if the arity of $\conn$ is 0, \iec if $\varphi$ is a logical constant, then $\defcon_{\conn}$ is some fixed element from $\Vm$).

If $\valu{I}{\Lm}(\varphi) \in \Vm^+$, then we say that $I$ is an \emph{$\mathcal{L}$-model} of $\varphi$, which we also denote by $I \models_{\mathcal{L}} \varphi$. An $\mathcal{L}$-formula $\varphi$ is called \emph{valid} iff every $\mathcal{L}$-interpretation of $\varphi$ is also an $\mathcal{L}$-model of $\varphi$. By $\mathit{Mod}_{\mathcal{L}}(\varphi)$ we denote the set of all $\mathcal{L}$-models of an $\mathcal{L}$-formula~$\varphi$.

By an \emph{$\Lm$-theory} we understand a set of $\Lm$-formulas.
An $\Lm$-interpretation $I$ is an $\Lm$-model of an $\Lm$-theory $\Gamma$ if $I$ is an $\Lm$-model of all elements of $\Gamma$. The set of all $\Lm$-models of a $\Lm$-theory $\Gamma$ is denoted by 
$\mathit{Mod}_{\mathcal{L}}(\Gamma)$.
An $\Lm$-formula $\varphi$ is a \emph{semantic consequence} of an $\Lm$-theory $\Gamma$ (in $\Lm$), denoted by 
$\Gamma \models_{\mathcal{L}} \varphi$, iff $\mathit{Mod}_{\mathcal{L}}(\Gamma) \subseteq \mathit{Mod}_{\mathcal{L}}(\varphi)$. Furthermore, for two $\mathcal{L}$-theories $\Gamma$ and $\Delta$, we define $\Gamma \models_{\mathcal{L}} \Delta$ iff  $\Gamma \models_{\mathcal{L}} \varphi$, for some $\varphi \in \Delta$.
 
 If it is clear from the context, to ease notation, we usually drop the prefix ``$\Lm$-'' in the concepts introduced above.

\paragraph{Three-Valued Paraconsistent Minimal Entailment.}

We define the three-valued paraconsistent entailment relation $\modelsLPm$, due to Priest~\cite{priest1991}, by means of the paraconsistent three-valued logic $\logicP$, following Avron~\cite{avron1991}.

The elements of $\logicP$ are as follows: (i)~the truth values of $\logicP$ are given by $\mathcal{V}_\logicP= \{\mathbf{f}, \mathbf{b}, \mathbf{t}\}$,  where $\mathbf{b}$ stands for ``both'', \iec the truth value referring to inconsistency; it is assumed that the truth values are ordered according to the stipulation that
$\mathbf{f} < \mathbf{b} < \mathbf{t}$;
(ii)~the designated truth values are $\mathcal{V}^+_{\logicP} = \{\mathbf{b}, \mathbf{t}\}$; (iii)~the primitive logical connectives of $\logicP$ are $\neg$, $\land$, and the 
logical constant $\CFalse$; and (iv) the valuation function $v^I_{\logicP}$, for an interpretation $I$, satisfies the following conditions:
\begin{itemize}
\item $v^I_{\logicP}(p) = I(p)$, for a propositional constant $p$;
\item $v^I_{\logicP}(\CFalse) = \mathbf{f}$;
\item $v^I_{\logicP}(\neg \varphi) = \mathbf{t}$ if $v^I_{\logicP}(\varphi) = \mathbf{f}$,
$v^I_{\logicP}(\neg \varphi) = \mathbf{f}$ if $v^I_{\logicP}(\varphi) = \mathbf{t}$, and
$v^I_{\logicP}(\neg \varphi) = \mathbf{b}$ if $v^I_{\logicP}(\varphi) = \mathbf{b}$;
\item $v^I_{\logicP}(\varphi \land \psi) = \mathit{min}(v^I_{\logicP}(\varphi),v^I_{\logicP}(\psi))$; and

\item $v^I_{\logicP}(\varphi \supset \psi) = v^I_{\logicP}(\psi)$ if $v^I_{\logicP}(\varphi) \in  \mathcal{V}^+_{\logicP}$, and
$v^I_{\logicP}(\varphi \supset \psi) = \mathbf{t}$ otherwise.
\end{itemize}

According to Avron~\cite{avron1999}, 
the connectives $\neg$, $\land$, $\supset$, and $\CFalse$ are functionally complete, i.e., any truth function (or, equivalently, logical connective) can be expressed by a $\logicP$-formula containing these connectives.
For example, the connective $\lor$ can be defined in the standard manner as $\varphi \lor \psi := \neg (\neg \varphi \land \neg \psi)$.

The \emph{logic of paradox}, $\logicLP$, due to Priest~\cite{priest1979},
is the sublogic of $\logicP$ obtained by excluding $\supset$ from the alphabet and using instead the defined implication $\varphi \rightarrow \psi:=\neg \varphi \lor \psi$.
For defining the relation $\modelsLPm$~\cite{priest1991}, let us call an
$\logicLP$-model $I$ of a theory $\Gamma$ \emph{minimally inconsistent} iff there is no other $\logicLP$-model $J$ of $\Gamma$ such that  
$\{ p \in \const \mid v^J_{{\logicLP}}(p) = \mathbf{b} \}\subset \{ p \in \const \mid v^I_{{\logicLP}}(p) = \mathbf{b} \}$.
Then, for theories $\Gamma$ and $\Delta$, $\Gamma\modelsLPm\Delta$ holds iff every 
minimally inconsistent $\logicLP$-model $I$ of $\Gamma$ is also a $\logicLP$-model of some $\varphi\in\Delta$. 
We also define analogously an entailment for $\logicP$, denoted by $\modelsPm$.

\paragraph{Four-Valued Paraconsistent Minimal Entailment.}
The four-valued paraconsistent minimal entailment relations $\modelsFourMa$ and $\modelsFourMb$, due to Arieli and Avron~\cite{arieli1998},
are defined in terms of the logic $\logicFour$ (also called $\mathbf{FOUR}$), which was introduced by Belnap~\cite{belnap1977b,belnap1977a} and extensively studied by Ginsberg~\cite{ginsberg1988}, Fitting~\cite{fitting1989,fitting1990}, and Arieli and Avron~\cite{arieli1994,arieli2000,arieli1996,arieli1998}.
Its truth values are  $\mathcal{V}_\logicFour = \{\mathbf{f},\mathbf{b}, \mathbf{n}, \mathbf{t}\}$, where $\mathbf{b}$ and $\mathbf{t}$ are designated, \iec  $\mathcal{V}_\logicFour^+ = \{\mathbf{b}, \mathbf{t}\}$, and $\mathbf{n}$ can be read as ``neither''. 
 The truth values of $\mathcal{V}_\logicFour$ are usually considered with respect to two partial orders: A  \emph{truth order}, $\leq_t$, and a \emph{knowledge order}, $\leq_k$. A simple way to depict both of those orders is to consider the truth values as elements of the bilattice as shown in Figure~\ref{fig:fourbilattice}, where $\leq_t$ is the order along the $x$-axis and $\leq_k$ the one along the $y$-axis.

\begin{figure}
\centering
\begin{tikzpicture}[scale=0.7]
  \node (top) at (0,0)   {$\mathbf{b}$};
  \node (f)   at (-1,-1) {$\mathbf{f}$};
  \node (t)   at (1,-1)   {$\mathbf{t}$};
  \node (bot) at (0,-2)   {$\mathbf{n}$};
  \draw (bot) -- (f) -- (top) -- (t) -- (bot);
  \node (corner1) at (-1.5,0.5) {};
  \node (corner2) at (-1.5,-2.5) {};
  \node (corner3) at (1.5,-2.5) {};
  \draw[->] (corner2) -- node[anchor=north] {$t$} (corner3);
  \draw[->] (corner2) -- node[anchor=east] {$k$} (corner1);
\end{tikzpicture}
\caption{The $\mathbf{FOUR}$ bilattice.}\label{fig:fourbilattice}
\end{figure}
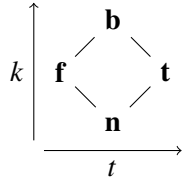
 
Following Arieli and Avron~\cite{arieli1998}, we take as primitive logical connectives of $\logicFour$ the operators $\neg, \land$, $\supset$,  and the logical constants $\CBoth$ and $\CNeither$.
Originally, $\supset$ is not part of the primitive connectives of $\logicFour$ but instead $\lor$ and the operators $\oplus$ and $\otimes$ are used, where the latter two work on the knowledge order rather than on the truth order as the other connectives. 
However, the set of connectives we use are functionally complete and $\lor$, $\oplus$, and $\otimes$ can thus be defined in terms of them.

The valuation function $v^I_{\logicFour}(\cdot)$ of $\logicFour$ is defined as follows:
\begin{itemize}
\item $v^I_{\logicFour}(p) = I(p)$, where $p$ is a propositional constant;
\item $v^I_{\logicFour}(\CBoth) = \mathbf{b}$ and $v^I_{\logicFour}(\CNeither) = \mathbf{n}$;
\item $v^I_{\logicFour}(\varphi \land \psi) = \mathit{min}_t(v^I_{\logicFour}(\varphi),v^I_{\logicFour}(\psi))$, where $\mathit{min}_t$ is the minimum with respect to $\leq_t$; 
\item $v^I_{\logicFour}(\neg \varphi) = \mathbf{t}$ if $v^I_{\logicFour}(\varphi) = \mathbf{f}$,
$v^I_{\logicFour}(\neg \varphi) = \mathbf{f}$ if $v^I_{\logicFour}(\varphi) = \mathbf{t}$, otherwise $v^I_{\logicFour}(\neg \varphi) =v^I_{\logicFour}(\varphi)$; and

\item 
$v^I_{\logicFour}(\varphi \supset \psi) = v^I_{\logicFour}(\psi)$ if $v^I_{\logicFour}(\varphi) \in \mathcal{V}^+_{\logicFour}$, and $v^I_{\logicFour}(\varphi \supset \psi) = \mathbf{t}$ otherwise.
\end{itemize}

From these conditions, we can define 
$\varphi \lor \psi := \neg (\neg \varphi \land \neg \psi)$,
$\varphi \otimes \psi:=(\varphi \land \mathbf{B}) \lor (\psi \land \mathbf{B}) \lor (\varphi \land \psi)$, and $\varphi \otimes \psi:=(\varphi \land \mathbf{N}) \lor (\psi \land \mathbf{N}) \lor (\varphi \land \psi)$.
It can easily be seen that $\land$ and $\lor$ correspond to the meet and join of the $\leq_t$-lattice whilst $\otimes$ and $\oplus$ correspond to the meet and join of the $\leq_k$-lattice.

The inference relations $\modelsFourMa$ and $\modelsFourMb$ by Arieli and Avron~\cite{arieli1998} are now defined thus:
Let us call an $\logicFour$-model $I$ of an $\logicFour$-theory $\Gamma$
\emph{most consistent} relative to a set $\mathcal{I}$ of truth values iff there is no other $\logicFour$-model $J$ such that
$\{ p \in \const \mid v^J_{{\logicFour}}(p) \in \mathcal{I} \} \subset \{ p \in \const \mid v^I_{{\logicFour}}(p) \in \mathcal{I} \}$.
Then, for $\logicFour$-theories $\Gamma$ and $\Delta$, 
$\Gamma \modelsFourMa \Delta$ holds iff
every $\logicFour$-model of $\Gamma$ which is most consistent relative to $\mathcal{I}=\{ \mathbf{b} \}$ is also an $\logicFour$-model of some formula in $\Delta$, while $\modelsFourMb$ is similarly defined but using $\mathcal{I} = \{ \mathbf{b}, \mathbf{n} \}$ instead.

\begin{example}\label{ex:modusPonens}
Consider $\Gamma = \{ p, \neg (p \land \neg q) \}$.
Then, $\Gamma \not\models_\logicP q$ as well as $\Gamma \not\models_\logicFour q$,
but $\Gamma \modelsPm q$, $\Gamma \modelsFourMa q$, and $\Gamma \modelsFourMb q$ all hold. Furthermore, for $\Gamma' = \Gamma \cup \{ \neg q\}$, we have $\Gamma' \not\modelsPm q$, $\Gamma' \not\modelsFourMa q$, and $\Gamma' \not\modelsFourMb q$. Hence, all those entailment relations are nonmonotonic.
Note also that $\Gamma'$ is clearly inconsistent in the sense of classical logic.
\end{example}

\section{Sequent Calculi for General Minimal Entailment}\label{sec:genMethod}

In order to obtain sequent-type calculi for the three- and four-valued paraconsistent entailment relations as defined above, we actually provide a uniform method for obtaining sequent calculi for generalised versions of these inference relations, given an arbitrary finite-valued logic as underlying base logic.
The calculi for $\modelsLPm$, $\modelsPm$, $\modelsFourMa$, and $\modelsFourMb$ are then obtained as special instances of the general method.

Following the sequent method of Bonatti and Olivetti~\cite{bonatti2002}, which we adopt here, our calculi involve three kinds of sequents, viz.\ assertional sequents for axiomatising validity in the underlying base logic, anti-sequents for axiomatising invalid formulas, and special sequents representing minimal entailment.

We start with defining our general minimal entailment relations and providing the postulates of the corresponding calculi, and afterwards we show soundness and completeness of the calculi.
The concrete systems for $\modelsLPm$, $\modelsPm$, $\modelsFourMa$, and $\modelsFourMb$ will be given in Section~\ref{sec:calcLP}.

Throughout this section, we assume to deal with a finite-valued logic $\Lm$ with truth values $\Vm=\{t_1\commadots t_n\}$ and a fixed set $\mathcal{I}\subseteq\Vm$ representing truth values to be minimised.
Our aim is to define a minimal entailment relation $\minModels$ and axiomatise it in terms of a sequent calculus.

Let us first define the relation $\minModels$.

\begin{definition}
Let $I$ and $J$ be $\Lm$-interpretations and $\Theta\subseteq\const$ be a set of propositional constants. Then, the relation
$I \leq_{\mathcal{L}}^{\mathcal{I},\Theta} J $ holds if 
$\{ p \in \Theta \mid v^I_{{\mathcal{L}}}(p) \in \mathcal{I} \} \subseteq \{ p \in \Theta\mid v^J_{{\mathcal{L}}}(p) \in \mathcal{I} \}$.
We write $I <_{\mathcal{L}}^{\mathcal{I},\Theta} J $ if $I \leq_{\mathcal{L}}^{\mathcal{I},\Theta} J $ but not $J \leq_{\mathcal{L}}^{\mathcal{I},\Theta} I $.

An $\Lm$-model of an $\Lm$-theory $\Gamma$ is \emph{$(\mathcal{I}; \Theta)$-minimal} if there is no $\Lm$-model $J$ of $\Gamma$ such that $J <_{\mathcal{L}}^{\mathcal{I},\Theta} I$.
If $\Theta=\const$, then an $(\mathcal{I}; \Theta)$-minimal model is simply referred to as being  \emph{$\mathcal{I}$-minimal}.

For $\mathcal{L}$-theories $\Gamma$ and $\Delta$, the relation
$\Gamma \minModels \Delta$ holds if for every $\mathcal{I}$-minimal $\mathcal{L}$-model $I$ of $\Gamma$, $I$ is an $\Lm$-model of some $\varphi \in \Delta$.
\end{definition}

In the context of relation $\minModels$, $\Lm$ is also referred to as the \emph{inner logic}.
Clearly, we have that
$\modelsLPm \, = \ \models_{\logicLP}^{\{\mathbf{b}\}}$, 
$\modelsPm  \, = \ \models_{\logicP}^{\{\mathbf{b}\}}$, 
$\modelsFourMa \, = \ \models_{\logicFour}^{\{\mathbf{b}\}}$, and 
$\modelsFourMb \, = \ \models_{\logicFour}^{\{\mathbf{b},\mathbf{n}\}}$.

As a first step towards our calculi, we now need sequent calculi for the inner logic $\Lm$ axiomatising, on the one hand, consequence $\Gamma\models_{\mathcal{L}}\Delta$ and, on the other hand, non-consequence $\Gamma\not\models_{\mathcal{L}}\Delta$.

For axiomatising consequence in $\Lm$, we use the method of Zach~\cite{zach1993}, who gave a general construction for obtaining sequent systems for any finite-valued logic, and for axiomatising non-consequence, we use the anti-sequent method of Bogojeski and Tompits~\cite{bogo-tompits20}, who provided a similar systematic method to obtain rejection systems for any finite-valued logic based on the method of Zach.
These methods use \emph{many-sided sequents}, following the original proposal of Rousseau~\cite{rousseau1967}, which is a natural generalisation for many-valued logics of the two-sided sequent method originally proposed by Gentzen~\cite{gentzen1935} for classical and intuitionistic logic. More specifically, both approaches reduce many-valued logics to two valued logic based on the concept of so-called \emph{partial normal forms}. Intuitively, those partial normal forms encode the many-valued semantics of the logical connectives into classical (two-valued) propositional formulas. From those normal forms, the needed rules for the connectives can then be derived.

For the purposes of axiomatising $\minModels$, it is not necessary at this point to fully specify the postulates of the calculi for $\Lm$, we only need to assume that such calculi exist---concrete systems for $\logicP$, $\logicLP$, and $\logicFour$ will be given in Section~\ref{sec:calcLP}.
We provide the necessary details in the following.

\begin{definition}
	An $\mathcal{L}$-\emph{sequent} for an $n$-valued logic $\mathcal{L}$ is an $n$-tuple $\mathfrak{S} = \Gamma_{1} \mid \dots \mid \Gamma_{n}$, where each $\Gamma_{i}$ is a finite set of $\Lm$-formulas, called \emph{component} of the sequent, and is associated with a truth value $t_i \in \mathcal{V}_\mathcal{L}$. For an $\mathcal{L}$-interpretation $I$, a sequent $\mathfrak{S}$ is \emph{true} under $I$ if some component $\Gamma_{t_i}$ contains some formula $\varphi$ such that $v_\mathcal{L}(\varphi) = t_i$. Furthermore, a sequent is \emph{valid} if it is true under any interpretation.
\end{definition}

Note that a standard sequent $\Gamma\vdash\Delta$ of classical logic in the sense of Gentzen~\cite{gentzen1935} corresponds to the sequent $\Gamma \mid\Delta$ according to the above definition.

As customary, we write sequent components comprised of a singleton set $\{\varphi\}$ simply as ``$\varphi$'' and similarly $\Gamma \cup \{\varphi\}$ as ``$\Gamma, \varphi$''.

Let us denote the sequent-type system for $\Lm$ based on $\Lm$-sequents obtained from the method of Zach~\cite{zach1993} by $\mathsf{S}_\mathcal{L}$. 
As these calculi do not encode logical consequence directly, but rather formalise truth conditions, we need some further notation.

First of all, by $\mathfrak{E}_n$ we denote the $\Lm$-sequent $\emptyset \mid \dots \mid \emptyset$.
Moreover, for two $\Lm$-sequents $\mathfrak{S}_1=\Gamma_1\mid\dots\mid\Gamma_n$ and $\mathfrak{S}_2=\Delta_1\mid\dots\mid\Delta_n$, we define the combination of $\mathfrak{S}_1$ and $\mathfrak{S}_2$ by $\mathfrak{S}_1,\mathfrak{S}_2:=\Gamma_1,\Delta_1\mid\dots\mid\Gamma_n, \Delta_n$. 

For a sequent $\mathfrak{S} = \Gamma_1\mid\cdots\mid\Gamma_n$ and a set $\Delta$ of formulas, $\mathfrak{S},[i:\Delta]$ denotes the $\Lm$-sequent that has the same components 
as $\mathfrak{S}$ but additionally contains $\Delta$ in its $i$-th component, i.e., 
$\mathfrak{S},[i:\Delta]=\Gamma_1\mid\cdots\mid\Gamma_i,\Delta \mid\cdots\mid\Gamma_n.$
This notation can also be applied repeatedly to a sequent in the following manner: Let 
$\mathfrak{S}=\Gamma_1\mid\cdots\mid\Gamma_n$, then 
$\mathfrak{S},[i_1:\Delta_1],\dots,[i_m:\Delta_m] := \Gamma_1\mid\cdots\mid\Gamma_{i_1},\Delta_1 
\mid\cdots\mid\Gamma_{i_m},\Delta_m\mid\cdots\mid\Gamma_n.$

Given an $\Lm$-sequent $\mathfrak{S}$, a set $\Delta$ of formulas, and a set $M\subseteq\{1,\dots,n\}$, we define
$\mathfrak{S},[M:\Delta]:=\mathfrak{S}, [i_1:\Delta], \dots, [i_m:\Delta],$
where $M=\{i_1,\dots,i_n\}$. 
For example, consider the three-component sequent
$\mathfrak{S}=\Gamma_1\mid\Gamma_2\mid\Gamma_3$, an arbitrary set $\Delta$ of formulas, and 
$M=\{1,3\}$. Then, 
$\mathfrak{S},[M:\Delta]=\mathfrak{S}, [1:\Delta], [3:\Delta] = \Gamma_1,\Delta\mid\Gamma_2\mid\Gamma_3,\Delta$.

\begin{definition}\label{def:consequence}
Let $\Gamma$ and $\Delta$ be $\mathcal{L}$-theories. Then, by $\Gamma \vdash_\mathcal{L} \Delta$ we denote the $\Lm$-sequent $\mathfrak{E}_n,[M^-:\Gamma],[M^+:\Delta]$, where $M^+=\{i\mid t_i\in \mathcal{V}_\mathcal{L}^+\}$ and $M^-=\{i\mid t_i\in \mathcal{V}_\mathcal{L} \setminus \mathcal{V}_\mathcal{L}^+ \}$.	
\end{definition}
Note that for, \egc $\mathcal{L}=\logicP$, $\Gamma \vdash_\logicP \Delta$ denotes the $\logicP$-sequent $\Gamma\mid\Delta\mid\Delta$.

The following result was shown by Zach~\cite{zach1993}:
\begin{proposition}\label{thm:compStandard}
$\Gamma \vdash_\mathcal{L} \Delta$ is provable in $\mathsf{S}_\mathcal{L}$ iff $\Gamma \models_\mathcal{L} \Delta$.
\end{proposition}

Now we provide the necessary details of the method of Bogojeski and Tompits~\cite{bogo-tompits20}.

\begin{definition}
	An $\mathcal{L}$-\emph{anti-sequent} for an $n$-valued logic $\mathcal{L}$ is an $n$-tuple ${\mathfrak{A}} = \Gamma_{1} \nmid \dots \nmid \Gamma_{n}$, where each $\Gamma_{i}$ is a finite set of $\Lm$-formulas, again called \emph{component} of the anti-sequent, and each component is associated with a truth value $t_i \in \mathcal{V}_\mathcal{L}$. For an $\mathcal{L}$-interpretation $I$, an anti-sequent ${\mathfrak{A}}$ is \emph{refuted} by $I$ if no component $\Gamma_{i}$ contains some formula $\varphi$ such that $v_\mathcal{L}(\varphi) = t_i$. Furthermore, an $\Lm$-anti-sequent is \emph{refutable} if it is refuted by some interpretation.
\end{definition}
Clearly, an $\mathcal{L}$-anti-sequent  $\Gamma_{1} \nmid \dots \nmid \Gamma_{n}$
is refutable iff the corresponding $\mathcal{L}$-sequent $\Gamma_{1} \mid \dots \mid \Gamma_{n}$ is valid.

Let us denote the anti-sequent calculus for $\Lm$ based on $\Lm$-anti-sequents obtained from the method of Bogojeski and Tompits~\cite{bogo-tompits20} by 
$\mathsf{R}_\mathcal{L}$.
Furthermore, the notation for combining $\Lm$-sequents is defined \emph{mutatis mutandis} for $\Lm$-anti-sequents, where, instead of ${\mathfrak{E}}_n$, we use the $\Lm$-anti-sequent 
${\mathfrak{F}}_n:=\emptyset \nmid \dots \nmid \emptyset$.

We next give the pendants of Definition~\ref{def:consequence} and Proposition~\ref{thm:compStandard}:
\begin{definition}
Let 
$\Gamma$ and $\Delta$ be $\mathcal{L}$-theories. Then, by $\Gamma \dashv_\mathcal{L} \Delta$ we denote the $\Lm$-anti-sequent ${\mathfrak{F}}_n,[M^-:\Gamma],[M^+:\Delta]$, where $M^+=\{i\mid t_i\in \mathcal{V}_\mathcal{L}^+\}$ and $M^-=\{i\mid t_i\in \mathcal{V}_\mathcal{L} \setminus \mathcal{V}_\mathcal{L}^+ \}$.
\end{definition}

\begin{proposition}[\cite{bogo-tompits20}]
$\Gamma \dashv_\mathcal{L} \Delta$ is provable in $\mathsf{R}_\mathcal{L}$ iff $\Gamma \not\models_\mathcal{L} \Delta$.
\end{proposition}

We are now in a position to define sequents capturing minimal entailment.

\begin{definition}
An $\MEL$-\emph{sequent} is defined as a quadruple of the form $\Sigma ; \Gamma \minSequent \Delta ; \Theta$, where $\Sigma, \Theta \subseteq \const$, and $\Gamma$ and $\Delta$ are $\mathcal{L}$-theories.

An $\MEL$-sequent $\Sigma ; \Gamma \minSequent \Delta ; \Theta$ is \emph{true} if, for every $\mathcal{L}$-interpretation $I$, 
if $I$ is an $(\mathcal{I};\Theta \cup \Sigma)$-minimal $\mathcal{L}$-model of $\Gamma$ such that for all $\psi \in \Sigma$, $v^I_{\mathcal{L}}(\psi) \in \mathcal{I}$ holds, then $I$ is an $\Lm$-model of some $\varphi \in \Delta$.
\end{definition} 

The connection between $\MEL$-sequents and the consequence relation $\models^{\mathcal{I}}_\mathcal{L}$ is established through the following theorem, whose proof is straightforward.
\begin{theorem}
Let $\Gamma$ and $\Delta$ be $\mathcal{L}$-theories. Then, $\Gamma \models^{\mathcal{I}}_\mathcal{L} \Delta$ iff $\emptyset ; \Gamma \Rightarrow_{\mathcal{L}}^{\mathcal{I}} \Delta ; \mathit{Var}(\Gamma \cup \Delta)$ is true, where $\mathit{Var}(\Gamma \cup \Delta)$ is the set of all propositional constants appearing in $\Gamma$ or $\Delta$.
\end{theorem}
We need one final definition towards defining our sequent systems for generalised minimal entailment:
\begin{definition}
For a many-valued logic $\mathcal{L}$, a set $\mathcal{I}\subseteq\Vm$ of truth values, and a truth-value $t \in \mathcal{V}_\mathcal{L}$, let $\unaryOp$ be a unary connective 
such that $v_\mathcal{L}(\unaryOp \ \varphi) = \mathbf{t}$ if $v_\mathcal{L}(\varphi) \in \mathcal{I}$, and $v_\mathcal{L}(\unaryOp \ \varphi) = \mathbf{f}$ otherwise. 
Furthermore, let $\unaryNotOp$ be the complementary connective such that $v_\mathcal{L}(\unaryNotOp \ \varphi) = \mathbf{t}$ if $v_\mathcal{L}(\varphi) \not\in \mathcal{I}$, and $v_\mathcal{L}(\unaryNotOp \ \varphi) = \mathbf{f}$ otherwise. 
Moreover, for a set $\Theta$ of propositional constants, let $\unaryOp  \Theta = \{  \unaryOp  p  \mid  p \in \Theta  \}$ and $\unaryNotOp \Theta = \{  \unaryNotOp  p  \mid  p \in \Theta  \}$.
\end{definition}

The motivation behind those connectives is that we want to be able fix the truth values of propositional constants. For example, if we require for an interpretation $I$ that it is a model of $\unaryOp  p$, then $p$ has to evaluate in $I$ to a truth value in $\mathcal{I}$.
We will assume that any of our inner logics contains such connectives, as we can always obtain corresponding rules for them in $\mathsf{S}_\mathcal{L}$ and $\mathsf{R}_\mathcal{L}$ using the constructions of Zach~\cite{zach1993} and Bogojeski and Tompits~\cite{bogo-tompits20}.

Having laid down the necessary concepts, we can now introduce the sequent-type calculus for minimal entailment.

\begin{definition}
The postulates of the calculus $\mathsf{ME}^\mathcal{I}_\mathcal{L}$ for minimal entailment consists of the postulates for the sequent calculus  
$\mathsf{S}_\mathcal{L}$, the postulates for the anti-sequent calculus $\mathsf{R}_\mathcal{L}$, and the additional inference rules for $\MEL$-sequents depicted in Figure~\ref{fig:genCalc}. 
\end{definition}

\begin{figure}[t!]
\hrule
\medskip
  \begin{center}
  \begin{minipage}[t]{0.40\textwidth}
\begin{prooftree}
\def\fCenter{\mbox{\ $\minSequent$\ }}
\AxiomC{$\Gamma, \unaryNotOp \ \Theta \dashv_\mathcal{L} \unaryOp \ q$}
\RightLabel{$(m_1)$}
\UnaryInfC{$q,\Sigma ; \Gamma \minSequent \Delta ; \Theta$}
\end{prooftree}
\end{minipage}
\begin{minipage}[t]{0.45\textwidth}
\begin{prooftree}
\def\fCenter{\mbox{\ $\minSequent$\ }}
\AxiomC{$\unaryOp \ \Sigma , \Gamma \vdash_\mathcal{L} \Delta$}
\RightLabel{$(m_2)$}
\UnaryInfC{$\Sigma ; \Gamma \minSequent \Delta ; \Theta$}
\end{prooftree}
\end{minipage}

\medskip

\begin{prooftree}
\def\fCenter{\mbox{\ $\minSequent$\ }}
\AxiomC{$q,\Sigma ; \Gamma \minSequent \Delta; \Theta$}
\AxiomC{$\Sigma ; \Gamma, \unaryNotOp \ q \minSequent \Delta; \Theta$}
\RightLabel{$(m_3)$}
\BinaryInfC{$\Sigma ; \Gamma \fCenter \Delta ; \Theta, q$}
\end{prooftree}

  \end{center}
\medskip  
where $\unaryNotOp \Theta = \{  \unaryNotOp  p \mid p \in \Theta \}$ and $\unaryOp  \Sigma = \{  \unaryOp  p \mid p \in \Sigma \}$ 
\vspace{0.8em}
\hrule
\caption{Additional rules of the sequent calculus $\mathsf{ME}^\mathcal{I}_\mathcal{L}$.}\label{fig:genCalc}
\end{figure}

The intuitive meaning of the inference rules $(m_1)$, $(m_2)$, and $(m_3)$ of Figure~\ref{fig:genCalc} is as follows:
If the premiss of rule $(m_1)$ is true, then there exists a model $I$ of $\Gamma$ where all elements of $\Theta$ and $q$ evaluate to truth values not in $\mathcal{L}$ under $I$. The model $I$ is clearly $(\mathcal{I} ; \Theta \cup \Sigma \cup \{q\})$-minimal and thus every model of $\Gamma$ where all elements of $\Theta$ and $q$ evaluate to truth values in $\mathcal{I}$ cannot be minimal. Hence, the sequent in the conclusion is vacuously true.
Rule $(m_2)$ basically states that consequences of the inner logic $\mathcal{L}$ are preserved under minimal entailment.
Lastly, rule $(m_3)$ allows to infer an $\MEL$-{sequent} by case distinction: the left premiss ensures that $\Delta$ holds in every $(\mathcal{I} ; \Theta \cup \Sigma)$-minimal model of $\Gamma$ in which $q$ evaluates to a truth value in $\mathcal{L}$, and the right premiss states that $\Delta$ holds in every $(\mathcal{I} ; \Theta \cup \Sigma)$-minimal model of $\Gamma$ in which $q$ does not evaluate to a truth value in $\mathcal{I}$. Thus, $q$ can be safely added to the set of constants to be minimised.

We next show the adequacy of our calculus. We start with the soundness of 
$\mathsf{ME}^\mathcal{I}_\mathcal{L}$.
\begin{theorem}[Soundness]
	If $\Sigma ; \Gamma \minSequent \Delta ; \Theta$ is provable in $\mathsf{ME}^\mathcal{I}_\mathcal{L}$, then it is true.
\end{theorem}
\begin{proof}
	The proof proceeds by showing the correctness of each rule.
	
	 We start with rule $(m_1)$. Suppose (i) its premiss $\Gamma, \unaryNotOp  \Theta \dashv_\mathcal{L} \unaryOp  q$ is refutable but (ii)~its conclusion $q,\Sigma ; \Gamma \minSequent \Delta ; \Theta$ is not true. By~(ii), there is an $(\mathcal{I} ; \Theta \cup \Sigma \cup \{q\})$-minimal model $I$ of $\Gamma$ such that for all $p \in \Sigma \cup \{q\}$, $v_\mathcal{L}(p) \in \mathcal{I}$. 
	 Similarly, by (i), there exists a model $J$ of $\Gamma \cup \unaryNotOp \ \Theta$ such that $v^J_\mathcal{L}(\unaryOp \ q) \not\in \mathcal{V}_\mathcal{L}^+$, or, equivalently, $v^J_\mathcal{L}(q) \not\in \mathcal{I}$. Trivially, $J$ is a model of $\Gamma$ and of $\unaryNotOp \ \Theta$, and $J \models_\Lm \unaryNotOp \ \Theta$ implies for all $p \in \Theta$, $v^J_\mathcal{L}(p) \not\in \mathcal{I}$. Now, since all elements of $\Sigma \cup \{q\}$ evaluate to a truth value in $\mathcal{I}$ under $I$ and all elements of $\Theta \cup \{q\}$ do not evaluate to a truth value in $\mathcal{I}$ under $J$, we have $J \leq^{\mathcal{I};{\Theta \cup \Sigma \cup \{q\}}}_\mathcal{L} I$. Furthermore, since $v^J_\mathcal{L}(q) \not\in \mathcal{I}$ but $v^I_\mathcal{L}(q) \in \mathcal{I}$, it even holds that $J <^{\mathcal{I};{\Theta \cup \Sigma \cup \{q\}}}_\mathcal{L} I$, which contradicts that $I$ is an $(\mathcal{I};\Theta \cup \Sigma \cup \{q\})$-minimal model of $\Gamma$. Hence, (ii) cannot be the case and the rule is indeed correct.

	 The correctness of rule $(m_2)$ is immediate, since any $(\mathcal{I};\Theta \cup \Sigma)$-minimal model $I$ of $\Gamma$ for which all elements of $\Sigma$ evaluate to truth values in $\mathcal{I}$ is trivially an $\mathcal{L}$-model of $\Gamma$. From the sequent in the premiss, it then follows that $I \models_\mathcal{L} \varphi$, for some $\varphi \in \Delta$. Hence, $\Sigma ; \Gamma \minSequent \Delta ; \Theta$ is true.
	 
	 To show the soundness of rule $(m_3)$, suppose that both sequents in its premiss are true. Furthermore, consider an $(\mathcal{I};\Theta\cup \{q\} \cup \Sigma )$-minimal model $I$ of $\Gamma$ where all elements of $\Sigma$ evaluate to truth values in $\mathcal{I}$.  We distinguish two cases: either
		(i) $v^I_\mathcal{L}(q) \in \mathcal{I}$ or (ii)~$v^I_\mathcal{L}(q) \not\in \mathcal{I}$. Suppose (i) holds. Then, $I$ is a $(\Theta \cup \Sigma \cup \{q\})$-minimal model of $\Gamma$ where all elements of $\Sigma \cup \{q\}$ evaluate to truth values in $\mathcal{I}$. Since $I$ is an $(\mathcal{I};\Theta \cup \Sigma \cup \{q\})$-minimal model and $v_\mathcal{L}^I(q) \in \mathcal{I}$, it is also an $(\mathcal{I};\Theta \cup \Sigma)$-minimal model of $\Gamma$, and because $q,\Sigma ; \Gamma \minSequent \Delta; \Theta$ holds, $I \models_\mathcal{L} \varphi$ follows, for some $\varphi \in \Sigma$. So, in case of~(i), the conclusion of the rule is true.
		
		It remains to consider case~(ii). Since $q$ does not evaluate to any truth value in $\mathcal{I}$, $I \models_\mathcal{L} \unaryNotOp \ q$ holds by definition. So, $I$ is an $(\mathcal{I};\Theta \cup \Sigma \cup \{q\})$-minimal model of $\Gamma \cup \{\neg \mathbf{I} q\}$ such that all elements of $\Sigma$ evaluate to truth values in $\mathcal{I}$, and thus also an $(\mathcal{I};\Theta \cup \Sigma)$-minimal model. Since $\Sigma ; \Gamma, \unaryNotOp \ q \minSequent \Delta; \Theta$ is true, it follows that for some $\varphi \in \Delta$, $I \models_\mathcal{L} \varphi$. Therefore, the conclusion of the rule also holds.
\end{proof}

\begin{theorem}[Completeness]
If $\Sigma ; \Gamma \minSequent \Delta ; \Theta$ is true, then it is provable in $\mathsf{ME}^\mathcal{I}_\mathcal{L}$.
\end{theorem}
\begin{proof}
	Suppose $S = \Sigma ; \Gamma \minSequent \Delta ; \Theta$ is true. We show the result by induction on $|\Theta|$.
	
\smallskip\noindent
{\sc Induction Base.}
	Assume $|\Theta| = 0$, \iec $\Theta = \emptyset$. If there is some $q \in \Sigma$ such that $\Gamma \dashv_\mathcal{L} \unaryOp \ q$ is refutable, then $S$ is provable by a single application of rule $(m_1)$. So, suppose that $\Gamma \vdash_\mathcal{L} \unaryOp \ q$ is valid, for any $q \in \Sigma$. Then, any model $I$ of $\Sigma \cup \Gamma$ is an $(\mathcal{I};\Theta \cup \Sigma)$-minimal model, since $\Theta = \emptyset$ and, by assumption, all elements in $\Sigma$ have to evaluate to truth values in $\mathcal{I}$ in every model of $\Gamma$. Now, since $I$ is minimal and $S$ is true by hypothesis, $I \models_\mathcal{L} \varphi$, for some $\varphi \in \Delta$. Hence, the $\mathcal{L}$-sequent $\unaryOp \ \Sigma, \Gamma \vdash_\mathcal{L} \Delta$ is valid and thus provable by the completeness of $\mathsf{S}_\Lm$.
A single application of rule $(m_2)$ then yields a proof of $S$.
	
\smallskip\noindent
{\sc Induction Step.}	Assume $|\Theta| > 0$ and that all true $\MEL$-sequents $\Sigma' ; \Gamma' \minSequent \Delta' ; \Theta'$ with $|\Theta| = |\Theta'| + 1 $ are provable in $\mathsf{ME}^\mathcal{I}_\mathcal{L}$. 
Suppose that $\Theta = \Theta' \cup \{p\}$, for some propositional constant $p$ such that $p \not\in \Theta'$. We show that $S_1: =  p, \Sigma ; \Gamma \minSequent \Delta ; \Theta'$ and $S_2 :=  \Sigma ; \Gamma, \unaryNotOp \ p \minSequent \Delta ; \Theta'$ are both true and thus, by induction hypothesis, also provable in $\mathsf{ME}^\mathcal{I}_\mathcal{L}$. Let $I$ be any $(\Theta' \cup \Sigma \cup \{p\})$-minimal model of $\Gamma$ where all elements of $\Sigma \cup \{p\}$ evaluate to truth values in $\mathcal{I}$. Trivially, $I$ is an $(\mathcal{I};\Theta \cup \Sigma)$-minimal model of $\Gamma$ where all elements of $\Sigma$ evaluate to truth values in $\mathcal{I}$. Since $S$ is true, it follows that $I \models_\mathcal{L} \varphi$, for some $\varphi \in \Delta$. Hence, $S_1$ is also true. On the other hand, suppose $I$ is an $(\mathcal{I};\Theta' \cup \Sigma)$-minimal model of $\Gamma \cup \unaryNotOp \ p$ where all elements of $\Sigma$ evaluate to truth values in $\mathcal{I}$. 
Now, $I \models_\mathcal{L} \unaryNotOp \ p$ implies $v^I_\mathcal{L}(p) \not\in \mathcal{I}$, and therefore $I$ is trivially also a $(\mathcal{I};\Theta \cup \Sigma)$-minimal model.
On the other hand, the truth of $S$ implies $I \models_\mathcal{L} \varphi$, for some $\varphi \in \Delta$, and thus $S_2$ is true as well. Since $S_1$ and $S_2$ are both true, and thus provable in in $\mathsf{ME}^\mathcal{I}_\mathcal{L}$ by induction hypothesis, a single application of rule $(m_3)$ yields a proof for $S$.
\end{proof}

\section{Calculi for Three- and Four-Valued Paraconsistent Logics}\label{sec:calcLP}

\begin{figure}[t!]
\hrule
\medskip
\begin{center}

\begin{minipage}[t]{0.30\textwidth}
\begin{prooftree}
\AxiomC{$\Gamma\mid\Delta\mid\Pi,\varphi$}
\RightLabel{$(\neg : \mathbf{f})^\vdash$}
\UnaryInfC{$\Gamma,\neg \varphi \mid\Delta\mid\Pi$}
\end{prooftree}
\end{minipage}
\begin{minipage}[t]{0.30\textwidth}
\begin{prooftree}
\AxiomC{$\Gamma\mid\Delta,\varphi\mid\Pi$}
\RightLabel{$(\neg : \mathbf{b})^\vdash$}
\UnaryInfC{$\Gamma\mid\Delta,\neg \varphi\mid\Pi$}
\end{prooftree}
\end{minipage}
\begin{minipage}[t]{0.30\textwidth}
\begin{prooftree}
\AxiomC{$\Gamma,\varphi\mid\Delta\mid\Pi$}
\RightLabel{$(\neg : \mathbf{t})^\vdash$}
\UnaryInfC{$\Gamma\mid\Delta\mid\Pi,\neg \varphi$}
\end{prooftree}
\end{minipage}

\medskip

\begin{minipage}[t]{0.45\textwidth}
\begin{prooftree}
\AxiomC{$\Gamma, \varphi, \psi \mid \Delta \mid \Pi $}
\RightLabel{$(\land : \mathbf{f})^\vdash$}
\UnaryInfC{$\Gamma, \varphi \land \psi \mid \Delta \mid \Pi $}
\end{prooftree}
\end{minipage}
\begin{minipage}[t]{0.45\textwidth}
\begin{prooftree}
\AxiomC{$\Gamma \mid \Delta \mid \Pi, \varphi$}
\AxiomC{$\Gamma \mid \Delta \mid \Pi, \psi$}
\RightLabel{$(\land : \mathbf{t})^\vdash$}
\BinaryInfC{$\Gamma \mid \Delta \mid \Pi, \varphi \land \psi  $}
\end{prooftree}
\end{minipage}

\begin{prooftree}
\AxiomC{$\Gamma \mid \Delta, \varphi, \psi \mid \Pi $}
\AxiomC{$\Gamma \mid \Delta, \varphi \mid \Pi, \varphi $}
\AxiomC{$\Gamma \mid \Delta, \psi \mid \Pi, \psi $}
\RightLabel{$(\land : \mathbf{b})^\vdash$}
\TrinaryInfC{$\Gamma \mid \Delta, \varphi \land \psi \mid \Pi $}
\end{prooftree}

\begin{minipage}[t]{0.45\textwidth}
\begin{prooftree}
\AxiomC{$\Gamma \mid \Delta, \varphi \mid \Pi, \varphi $}
\AxiomC{$\Gamma, \psi \mid \Delta \mid \Pi $}
\RightLabel{$(\supset : \mathbf{f})^\vdash$}
\BinaryInfC{$\Gamma, \varphi \supset \psi \mid \Delta \mid \Pi $}
\end{prooftree}
\end{minipage}
\begin{minipage}[t]{0.45\textwidth}
\begin{prooftree}
\AxiomC{$\Gamma \mid \Delta, \varphi \mid \Pi, \varphi $}
\AxiomC{$\Gamma \mid \Delta, \psi \mid \Pi $}
\RightLabel{$(\supset : \mathbf{b})^\vdash$}
\BinaryInfC{$\Gamma\mid \Delta, \varphi \supset \psi  \mid \Pi $}
\end{prooftree}
\end{minipage}

\begin{prooftree}
\AxiomC{$\Gamma, \varphi \mid \Delta \mid \Pi, \psi $}
\RightLabel{$(\supset : \mathbf{t})^\vdash$}
\UnaryInfC{$\Gamma\mid \Delta \mid \Pi, \varphi \supset \psi  $}
\end{prooftree}

\begin{minipage}[t]{0.30\textwidth}
\begin{prooftree}
\AxiomC{$\Gamma\mid \Delta \mid \Pi $}
\RightLabel{$(w : \mathbf{f})^\vdash$}
\UnaryInfC{$\Gamma, \varphi \mid \Delta \mid \Pi $}
\end{prooftree}
\end{minipage}
\begin{minipage}[t]{0.30\textwidth}
\begin{prooftree}
\AxiomC{$\Gamma\mid \Delta \mid \Pi $}
\RightLabel{$(w : \mathbf{b})^\vdash$}
\UnaryInfC{$\Gamma\mid \Delta, \varphi  \mid \Pi $}
\end{prooftree}
\end{minipage}
\begin{minipage}[t]{0.30\textwidth}
\begin{prooftree}
\AxiomC{$\Gamma\mid \Delta \mid \Pi $}
\RightLabel{$(w : \mathbf{t})^\vdash$}
\UnaryInfC{$\Gamma\mid \Delta \mid \Pi, \varphi $}
\end{prooftree}
\end{minipage}

\medskip

\end{center}
\hrule
\caption{Rules of the sequent calculus $\mathsf{S}_\logicP$.}\label{fig:pmCalcPositive}
\end{figure}

From the results in the previous section, we can obtain now concrete calculi for 
axiomatising $\modelsLPm$, $\modelsPm$, $\modelsFourMa$, and $\modelsFourMb$.
We start with the three-valued case.
Since $\logicLP$ is a sublogic of $\logicP$, we only deal with the case of 
$\modelsPm$.

To begin with, following from the general construction of Zach~\cite{zach1993} and Bogojeski and Tompits~\cite{bogo-tompits20}, we obtain a sequent calculus $\mathsf{S}_\logicP$ and an anti-sequent calculus $\mathsf{R}_\logicP$ for $\logicP$ as follows: 

\begin{enumerate}[(i)]
\item the axioms of $\mathsf{S}_\logicP$ are $\logicP$-sequents of the form
$\Gamma_1,\varphi \mid \Gamma_2,\varphi \mid \Gamma_3,\varphi$ and
	$\Gamma_1, \CFalse \mid \Gamma_2 \mid \Gamma_3$, and
the inference rules of $\mathsf{S}_\logicP$ are those depicted in Figure~\ref{fig:pmCalcPositive}; and

\item the axioms of $\mathsf{R}_\logicP$ are $\logicP$-anti-sequents of the form $\Gamma_1 \nmid \Gamma_2 \nmid \Gamma_3$, where $\Gamma_1, \Gamma_2, \Gamma_3$ are sets of propositional and logical constants such that $\Gamma_1 \cap \Gamma_2 \cap \Gamma_3 = \emptyset$ and $\CFalse \not\in \Gamma_1$, and the inference rules of $\mathsf{R}_\logicP$ are those depicted in Figure~\ref{fig:pmCalcNegative}.

\end{enumerate}

Note that the inference rules of $\mathsf{S}_\logicP$ and $\mathsf{R}_\logicP$ contain only those for the primitive logical connectives. Furthermore, the rules 
$(w : \mathbf{f})^\vdash$, $(w : \mathbf{b})^\vdash$, and $(w : \mathbf{t})^\vdash$ are called \emph{weakening rules}.

The intuition behind the postulates of $\mathsf{S}_\logicP$ and $\mathsf{R}_\logicP$ is the following: 
An axiom of $\mathsf{S}_\logicP$ of the form $\Gamma_1,\varphi \mid \Gamma_2,\varphi \mid \Gamma_3,\varphi$ simply expresses the three-valuedness of the logic $\logicP$, \iec that any formula $\varphi$ must have one of the three truth values $\mathbf{f}$, $\mathbf{b}$, or $\mathbf{t}$, while an axiom of the form $\Gamma_1, \CFalse \mid \Gamma_2 \mid \Gamma_3$ is trivially valid because the truth constant $\CFalse$ is always false. 
The axioms of $\mathsf{R}_\logicP$, on the other hand, represent basically the complementary situation of atomic $\logicP$-sequents, encoding a refuting interpretation. 
As for the inference rules of both $\mathsf{S}_\logicP$ and $\mathsf{R}_\logicP$, they intuitively express the truth-table conditions of the different connectives obtained from a specification in two-valued logic. For instance, the rule $(\supset : \mathbf{f})^\vdash$ expresses the semantic conditions when an implication $\varphi\supset\psi$ is false, which is the case when $\varphi$ has one of the designated truth values $\mathbf{b}$ or $\mathbf{t}$, and $\psi$ is false.
Note that the rules of $\mathsf{R}_\logicP$ are always unary as they intuitively correspond to the branches of a systematic search for countermodels in the standard sequent calculus. 
Roughly speaking, what is exhaustive search in the standard calculus amounts to nondeterminism in the anti-sequent calculus.

From the general construction of Zach~\cite{zach1993} and Bogojeski and Tompits~\cite{bogo-tompits20}, it follows that $\mathsf{S}_\logicP$ and $\mathsf{R}_\logicP$ are sound and complete, \iec a $\logicP$-sequent $\Gamma_1 \mid \Gamma_2 \mid \Gamma_3$ is valid iff it is provable in $\mathsf{S}_\logicP$, and a $\logicP$-anti-sequent $\Gamma_1 \nmid \Gamma_2 \nmid \Gamma_3$ is refutable iff it is provable in $\mathsf{R}_\logicP$.  

The calculus $\mathsf{ME}_\logicP$ for $\modelsPm$ comprises now the calculi $\mathsf{S}_\logicP$ and $\mathsf{R}_\logicP$, and the inference rules for $\MEL$-sequents as described in Figure~\ref{fig:genCalc}, setting $\mathcal{I}=\{\mathbf{b}\}$ and $\Lm=\logicP$. However, instead of the general rules $(m_1)$ and $(m_2)$, we may use the following versions which directly encode the semantics of the operators 
$\unaryOpPm$ and $\unaryNotOpPm$, instead of providing explicit inference rules for them in the calculi $\mathsf{R}_\logicP$ and $\mathsf{R}_\logicP$:
\vspace{-1ex}
\begin{center} 
\begin{minipage}[t]{0.40\textwidth}
\begin{prooftree}
\def\fCenter{\mbox{\ $\Rightarrow^{\{ \mathbf{b} \}}_{\logicP}$\ }}
\AxiomC{$\Gamma \nmid \Theta, \Pi, q \nmid \emptyset$}
\RightLabel{$(m_1')$}
\UnaryInfC{$q,\Sigma ; \Gamma, \unaryNotOpPm \ \Pi \Rightarrow^{\{ \mathbf{b} \}}_{\logicP} \Delta ; \Theta$}
\end{prooftree}
\end{minipage}
\begin{minipage}[t]{0.45\textwidth}
\begin{prooftree}
\def\fCenter{\mbox{\ $\Rightarrow^{\{ \mathbf{b} \}}_{\logicP}$\ }}
\AxiomC{$\Sigma, \Gamma \mid \Delta, \Pi \mid \Delta, \Sigma$}
\RightLabel{$(m_2')$}
\UnaryInfC{$\Sigma ; \Gamma, \unaryNotOpPm \ \Pi \Rightarrow^{\{ \mathbf{b} \}}_{\logicP} \Delta ; \Theta$}
\end{prooftree}
\end{minipage}
\end{center}

\bigskip

\begin{figure}[t!]
\hrule
\medskip
\begin{center}

\begin{minipage}[t]{0.30\textwidth}
\begin{prooftree}
\AxiomC{$\Gamma\nmid\Delta\nmid\Pi,\varphi$}
\RightLabel{$(\neg : \mathbf{f})^\dashv$}
\UnaryInfC{$\Gamma,\neg \varphi \nmid\Delta\nmid\Pi$}
\end{prooftree}
\end{minipage}
\begin{minipage}[t]{0.30\textwidth}
\begin{prooftree}
\AxiomC{$\Gamma\nmid\Delta,\varphi\nmid\Pi$}
\RightLabel{$(\neg : \mathbf{b})^\dashv$}
\UnaryInfC{$\Gamma\nmid\Delta,\neg \varphi\nmid\Pi$}
\end{prooftree}
\end{minipage}
\begin{minipage}[t]{0.30\textwidth}
\begin{prooftree}
\AxiomC{$\Gamma,\varphi\nmid\Delta\nmid\Pi$}
\RightLabel{$(\neg : \mathbf{t})^\dashv$}
\UnaryInfC{$\Gamma\nmid\Delta\nmid\Pi,\neg \varphi$}
\end{prooftree}
\end{minipage}

\medskip

\begin{minipage}[t]{0.30\textwidth}
\begin{prooftree}
\AxiomC{$\Gamma, \varphi, \psi \nmid \Delta \nmid \Pi $}
\RightLabel{$(\land : \mathbf{f})^\dashv$}
\UnaryInfC{$\Gamma, \varphi \land \psi \nmid \Delta \nmid \Pi $}
\end{prooftree}
\end{minipage}
\begin{minipage}[t]{0.30\textwidth}
\begin{prooftree}
\AxiomC{$\Gamma \nmid \Delta, \varphi, \psi \nmid \Pi$}
\RightLabel{$(\land : \mathbf{b}^1)^\dashv$}
\UnaryInfC{$\Gamma \nmid \Delta, \varphi \land \psi \nmid \Pi $}
\end{prooftree}
\end{minipage}
\begin{minipage}[t]{0.30\textwidth}
\begin{prooftree}
\AxiomC{$\Gamma \nmid \Delta, \varphi \nmid \Pi, \varphi $}
\RightLabel{$(\land : \mathbf{b}^2)^\dashv$}
\UnaryInfC{$\Gamma \nmid \Delta, \varphi \land \psi \nmid \Pi $}
\end{prooftree}
\end{minipage}

\medskip

\begin{minipage}[t]{0.30\textwidth}
\begin{prooftree}
\AxiomC{$\Gamma \nmid \Delta, \psi \nmid \Pi, \psi $}
\RightLabel{$(\land : \mathbf{b}^3)^\dashv$}
\UnaryInfC{$\Gamma \nmid \Delta, \varphi \land \psi \nmid \Pi $}
\end{prooftree}
\end{minipage}
\begin{minipage}[t]{0.30\textwidth}
\begin{prooftree}
\AxiomC{$\Gamma \nmid \Delta, \nmid \Pi, \varphi$}
\RightLabel{$(\land : \mathbf{t}^1)^\dashv$}
\UnaryInfC{$\Gamma \nmid \Delta, \nmid \Pi, \varphi \land \psi $}
\end{prooftree}
\end{minipage}
\begin{minipage}[t]{0.30\textwidth}
\begin{prooftree}
\AxiomC{$\Gamma \nmid \Delta, \nmid \Pi, \psi $}
\RightLabel{$(\land : \mathbf{t}^2)^\dashv$}
\UnaryInfC{$\Gamma \nmid \Delta \nmid \Pi, \varphi \land \psi $}
\end{prooftree}
\end{minipage}

\medskip

\begin{minipage}[t]{0.45\textwidth}
\begin{prooftree}
\AxiomC{$\Gamma \nmid \Delta, \varphi \nmid \Pi, \varphi $}
\RightLabel{$(\supset : \mathbf{f}^1)^\dashv$}
\UnaryInfC{$\Gamma, \varphi \supset \psi \nmid \Delta \nmid \Pi $}
\end{prooftree}
\end{minipage}
\begin{minipage}[t]{0.45\textwidth}
\begin{prooftree}
\AxiomC{$\Gamma, \psi \nmid \Delta \nmid \Pi $}
\RightLabel{$(\supset : \mathbf{f}^2)^\dashv$}
\UnaryInfC{$\Gamma, \varphi \supset \psi \nmid \Delta \nmid \Pi $}
\end{prooftree}
\end{minipage}

\medskip

\begin{minipage}[t]{0.32\textwidth}
\begin{prooftree}
\AxiomC{$\Gamma \nmid \Delta, \varphi \nmid \Pi, \varphi $}
\RightLabel{$(\supset : \mathbf{b}^1)^\dashv$}
\UnaryInfC{$\Gamma \nmid \Delta, \varphi \supset \psi \nmid \Pi $}
\end{prooftree}
\end{minipage}
\begin{minipage}[t]{0.32\textwidth}
\begin{prooftree}
\AxiomC{$\Gamma \nmid \Delta, \psi \nmid \Pi $}
\RightLabel{$(\supset : \mathbf{b}^2)^\dashv$}
\UnaryInfC{$\Gamma\nmid \Delta, \varphi \supset \psi  \nmid \Pi $}
\end{prooftree}
\end{minipage}
\begin{minipage}[t]{0.32\textwidth}
\begin{prooftree}
\AxiomC{$\Gamma, \varphi \nmid \Delta \nmid \Pi, \psi $}
\RightLabel{$(\supset : \mathbf{t})^\dashv$}
\UnaryInfC{$\Gamma \nmid \Delta \nmid \Pi, \varphi \supset \psi $}
\end{prooftree}
\end{minipage}

\end{center}
\hrule
\caption{Rules of the anti-sequent calculus $\mathsf{R}_\logicP$.}\label{fig:pmCalcNegative}
\end{figure}

Following from our results in Section~\ref{sec:genMethod}, the calculus $\mathsf{ME}_\logicP$ is sound and complete. Hence, we get the following corollary:

\begin{theorem}
	Let $\Gamma$ and $\Delta$ be $\mathcal{P}$-theories. Then, $\Gamma \modelsPm \Delta$ 
iff $\emptyset ; \Gamma \Rightarrow \Delta ; \mathit{Var}(\Gamma \cup \Delta)$ is provable in $\mathsf{ME}_\logicP$, where $\mathit{Var}(\Gamma \cup \Delta)$ is the set of propositional constants appearing in $\Gamma$ or $\Delta$.
\end{theorem}
Note that, if $\Gamma$ and $\Delta$ do not contain $\supset$, then the above result holds also for $\modelsLPm$.  
\begin{example}\label{ex:smp}
Recall the theory $\Gamma = \{ p, \neg (p \land \neg q)\}$ from Example~\ref{ex:modusPonens}. 
As $\Gamma \modelsPm q$ holds, the sequent $\emptyset ; p, \neg (p \land \neg q) \Rightarrow^{\{ \mathbf{b} \}}_{\logicP} q ; p,q$ is provable in $\mathsf{ME}_\logicP$. A proof of the sequent is as follows:

\begin{prooftree}
\AxiomC{$p, q \nmid q, p \nmid \emptyset$}
\RightLabel{$(\neg : \mathbf{t})^\dashv$}
\UnaryInfC{$p \nmid q, p \nmid \neg q$}
\RightLabel{$(\land : \mathbf{t}^2)^\dashv$}
\UnaryInfC{$p \nmid q, p \nmid p \land \neg q$}
\RightLabel{$(\neg : \mathbf{f})^\dashv$}
\UnaryInfC{$p, \neg (p \land \neg q) \nmid q, p \nmid \emptyset$}
\RightLabel{$(m_1')$}
\UnaryInfC{$p ; p, \neg (p \land \neg q) \Rightarrow^{\{ \mathbf{b} \}}_{\logicP} q ; q$}
\AxiomC{$p \mid q, p \mid q,p$}
\AxiomC{$p, q \mid q, p \mid q$}
\RightLabel{$(\neg : \mathbf{t})^\vdash$}
\UnaryInfC{$p \mid q, p \mid q, \neg q$}
\RightLabel{$(\land : \mathbf{t})^\vdash$}
\BinaryInfC{$p \mid q,p \mid q,p \mid q, p \land \neg q$}
\RightLabel{$(\neg : \mathbf{f})^\vdash$}
\UnaryInfC{$p, \neg (p \land \neg q) \mid q,p \mid q,p \mid q$}
\RightLabel{$(m_2')$}
\UnaryInfC{$\emptyset ; p, \neg (p \land \neg q), \unaryNotOpPm p \Rightarrow^{\{ \mathbf{b} \}}_{\logicP} q ; q $}
\RightLabel{$(m_3)$}
\BinaryInfC{$\emptyset ; p, \neg (p \land \neg q) \Rightarrow^{\{ \mathbf{b} \}}_{\logicP} q ; p,q$}
\end{prooftree}
Note that the top-most sequents are axioms in $\mathsf{S}_\logicP$ and $\mathsf{R}_\logicP$, respectively. 	
\end{example}

\begin{figure}[t!]
\hrule
\medskip
\begin{center}

\begin{minipage}[t]{0.32\textwidth}
\begin{prooftree}
\AxiomC{$\Gamma\mid\Delta\mid\Pi\mid\Omega,\varphi$}
\RightLabel{$(\neg : \mathbf{f})^\vdash$}
\UnaryInfC{$\Gamma,\neg \varphi \mid\Delta\mid\Pi\mid\Omega$}
\end{prooftree}
\end{minipage}
\begin{minipage}[t]{0.32\textwidth}
\begin{prooftree}
\AxiomC{$\Gamma\mid\Delta,\varphi\mid\Pi\mid\Omega$}
\RightLabel{$(\neg : \mathbf{n})^\vdash$}
\UnaryInfC{$\Gamma\mid\Delta,\neg \varphi\mid\Pi\mid\Omega$}
\end{prooftree}
\end{minipage}
\begin{minipage}[t]{0.32\textwidth}
\begin{prooftree}
\AxiomC{$\Gamma\mid\Delta\mid\Pi,\varphi\mid\Omega$}
\RightLabel{$(\neg : \mathbf{b})^\vdash$}
\UnaryInfC{$\Gamma\mid\Delta\mid\Pi,\neg \varphi\mid\Omega$}
\end{prooftree}
\end{minipage}

\begin{prooftree}
\AxiomC{$\Gamma,\varphi\mid\Delta\mid\Pi\mid\Omega$}
\RightLabel{$(\neg : \mathbf{t})^\vdash$}
\UnaryInfC{$\Gamma\mid\Delta\mid\Pi\mid\Omega,\neg \varphi$}
\end{prooftree}

\begin{prooftree}
\AxiomC{$\Gamma, \varphi, \psi \mid \Delta, \varphi, \psi \mid \Pi \mid \Omega $}
\AxiomC{$\Gamma, \varphi, \psi \mid \Delta \mid \Pi, \varphi, \psi \mid \Omega $}
\RightLabel{$(\land : \mathbf{f})^\vdash$}
\BinaryInfC{$\Gamma, \varphi \land \psi \mid \Delta \mid \Pi \mid \Omega $}
\end{prooftree}

\begin{prooftree}
\AxiomC{$\Gamma \mid \Delta, \varphi, \psi \mid \Pi \mid \Omega $}
\AxiomC{$\Gamma \mid \Delta, \varphi \mid \Pi \mid \Omega, \varphi $}
\AxiomC{$\Gamma \mid \Delta, \psi \mid \Pi \mid \Omega, \psi $}
\RightLabel{$(\land : \mathbf{n})^\vdash$}
\TrinaryInfC{$\Gamma \mid \Delta, \varphi \land \psi \mid \Pi \mid \Omega $}
\end{prooftree}

\begin{prooftree}
\AxiomC{$\Gamma \mid \Delta \mid \Pi, \varphi, \psi \mid \Omega $}
\AxiomC{$\Gamma \mid \Delta \mid \Pi, \varphi \mid \Omega, \varphi $}
\AxiomC{$\Gamma \mid \Delta \mid \Pi, \psi \mid \Omega, \psi $}
\RightLabel{$(\land : \mathbf{b})^\vdash$}
\TrinaryInfC{$\Gamma \mid \Delta \mid \Pi, \varphi \land \psi \mid \Omega $}
\end{prooftree}

\begin{prooftree}
\AxiomC{$\Gamma \mid \Delta \mid \Pi \mid \Omega, \varphi $}
\AxiomC{$\Gamma \mid \Delta \mid \Pi \mid \Omega, \psi $}
\RightLabel{$(\land : \mathbf{t})^\vdash$}
\BinaryInfC{$\Gamma \mid \Delta \mid \Pi \mid\Omega, \varphi \land \psi $}
\end{prooftree}

\begin{minipage}[t]{0.496\textwidth}
\begin{prooftree}
\AxiomC{$\Gamma \mid \Delta \mid \Pi, \varphi \mid \Omega, \varphi$}
\AxiomC{$\Gamma, \psi \mid \Delta \mid \Pi \mid\Omega$}
\RightLabel{$(\supset : \mathbf{f})^\vdash$}
\BinaryInfC{$\Gamma, \varphi \supset \psi \mid \Delta \mid \Pi \mid\Omega $}
\end{prooftree}
\end{minipage}
\begin{minipage}[t]{0.496\textwidth}
\begin{prooftree}
\AxiomC{$\Gamma \mid \Delta \mid \Pi, \varphi \mid \Omega, \varphi$}
\AxiomC{$\Gamma \mid \Delta, \psi \mid \Pi \mid\Omega$}
\RightLabel{$(\supset : \mathbf{n})^\vdash$}
\BinaryInfC{$\Gamma\mid \Delta, \varphi \supset \psi  \mid \Pi \mid\Omega $}
\end{prooftree}
\end{minipage}

\medskip

\begin{minipage}[t]{0.49\textwidth}
\begin{prooftree}
\AxiomC{$\Gamma \mid \Delta \mid \Pi, \varphi \mid \Omega, \varphi$}
\AxiomC{$\Gamma \mid \Delta \mid \Pi, \psi \mid\Omega$}
\RightLabel{$(\supset : \mathbf{b})^\vdash$}
\BinaryInfC{$\Gamma\mid \Delta \mid \Pi , \varphi \supset \psi \mid\Omega $}
\end{prooftree}
\end{minipage}
\begin{minipage}[t]{0.49\textwidth}
\begin{prooftree}
\AxiomC{$\Gamma, \varphi \mid \Delta, \varphi \mid \Pi\mid\Omega, \psi $}
\RightLabel{$(\supset : \mathbf{t})^\vdash$}
\UnaryInfC{$\Gamma\mid \Delta \mid \Pi \mid\Omega, \varphi \supset \psi  $}
\end{prooftree}
\end{minipage}

\medskip

\begin{minipage}[t]{0.245\textwidth}
\begin{prooftree}
\AxiomC{$\Gamma\mid \Delta \mid \Pi \mid\Omega$}
\RightLabel{$(w : \mathbf{f})^\vdash$}
\UnaryInfC{$\Gamma, \varphi \mid \Delta \mid \Pi \mid\Omega $}
\end{prooftree}
\end{minipage}
\begin{minipage}[t]{0.245\textwidth}
\begin{prooftree}
\AxiomC{$\Gamma\mid \Delta \mid \Pi \mid\Omega$}
\RightLabel{$(w : \mathbf{n})^\vdash$}
\UnaryInfC{$\Gamma\mid \Delta, \varphi  \mid \Pi \mid\Omega $}
\end{prooftree}
\end{minipage}
\begin{minipage}[t]{0.245\textwidth}
\begin{prooftree}
\AxiomC{$\Gamma\mid \Delta \mid \Pi \mid\Omega$}
\RightLabel{$(w : \mathbf{b})^\vdash$}
\UnaryInfC{$\Gamma\mid \Delta \mid \Pi, \varphi \mid\Omega $}
\end{prooftree}
\end{minipage}
\begin{minipage}[t]{0.245\textwidth}
\begin{prooftree}
\AxiomC{$\Gamma\mid \Delta \mid \Pi \mid\Omega$}
\RightLabel{$(w : \mathbf{t})^\vdash$}
\UnaryInfC{$\Gamma\mid \Delta \mid \Pi \mid\Omega, \varphi $}
\end{prooftree}
\end{minipage}

\end{center}

\vspace{0.8em}
\hrule
\caption{Rules of the sequent calculus $\mathsf{S}_\logicFour$.}\label{fig:fCalcPositive}
\end{figure}

Let us now consider the calculi for
$\modelsFourMa$ and $\modelsFourMb$.
For the inner logic $\logicFour$, we obtain the calculi $\mathsf{S}_\logicFour$ and $\mathsf{R}_\logicFour$ as follows:

\begin{enumerate}[(i)]
	\item the axioms of $\mathsf{S}_\logicFour$ are $\logicFour$-sequents of the form 
	$$\Gamma_1,\varphi \mid \Gamma_2,\varphi \mid \Gamma_3,\varphi \mid \Gamma_4,\varphi, \quad 
\Gamma_1 \mid \Gamma_2, \CNeither \mid \Gamma_3 \mid \Gamma_4, \quad \mbox{and} \quad
\Gamma_1 \mid \Gamma_2 \mid \Gamma_3, \CBoth \mid \Gamma_4,$$
and the inference rules of $\mathsf{S}_\logicFour$ are those given in Figure~\ref{fig:fCalcPositive}, and

	\item the axioms of $\mathsf{R}_\logicFour$ are $\logicFour$-anti-sequents of the form $\Gamma_1 \nmid \Gamma_2 \nmid \Gamma_3 \nmid \Gamma_4$, where $\Gamma_1$, $\Gamma_2$, $\Gamma_3$, and $\Gamma_4$ are sets of propositional and logical constants such that $\Gamma_1 \cap \Gamma_2 \cap \Gamma_3 \cap \Gamma_4 = \emptyset$, $\CNeither \not\in \Gamma_2$, and $\CBoth \not\in \Gamma_3$, and the inference rules of $\mathsf{R}_\logicFour$ are given in Figure~\ref{fig:fCalcNegative}.
\end{enumerate}

Again, these calculi are sound and complete and the intuition behind the axioms and rules is similar to that of the postulates of $\mathsf{S}_\logicP$ and $\mathsf{R}_\logicP$, respectively.
Also, for the calculi $\mathsf{ME}_{\logicFour}^1$ for $\modelsFourMa$ and 
$\mathsf{ME}_{\logicFour}^2$ for $\modelsFourMb$, which includes the calculi $\mathsf{S}_\logicFour$ and $\mathsf{R}_\logicFour$, we use the instance of rule $(m_3)$ for the logics at hand and variants of rules $(m_1)$ and $(m_2)$ which again directly encode the semantic properties of the operators $\unaryOp$ and $\unaryNotOp$ as follows: for  
$\mathsf{ME}_{\logicFour}^1$, we use the rules
\begin{center} 
\begin{minipage}[t]{0.40\textwidth}
\begin{prooftree}
\def\fCenter{\mbox{\ $\Rightarrow^{\{ \mathbf{b} \}}_{\logicFour}$\ }}
\AxiomC{$\Gamma \nmid \Gamma \nmid \Theta, \Pi, q \nmid \emptyset$}
\RightLabel{$(m_1^\dagger)$}
\UnaryInfC{$q,\Sigma ; \Gamma, \unaryNotOpFa \ \Pi \Rightarrow^{\{ \mathbf{b} \}}_{\logicFour} \Delta ; \Theta$}
\end{prooftree}
\end{minipage}
\begin{minipage}[t]{0.45\textwidth}
\begin{prooftree}
\def\fCenter{\mbox{\ $\Rightarrow^{\{ \mathbf{b} \}}_{\logicFour}$\ }}
\AxiomC{$\Sigma, \Gamma \mid \Sigma, \Gamma \mid \Pi, \Delta \mid \Sigma, \Delta$}
\RightLabel{$(m_2^\dagger)$}
\UnaryInfC{$\Sigma ; \Gamma, \unaryNotOpFa \ \Pi \Rightarrow^{\{ \mathbf{b} \}}_{\logicFour} \Delta ; \Theta$}
\end{prooftree}
\end{minipage}
\end{center}
\bigskip
and for $\mathsf{ME}_{\logicFour}^2$, we use the rules
\vspace{-1ex}
\begin{center} 
\begin{minipage}[t]{0.40\textwidth}
\begin{prooftree}
\def\fCenter{\mbox{\ $\Rightarrow^{\{ \mathbf{b}, \mathbf{n} \}}_{\logicFour}$\ }}
\AxiomC{$\Gamma \nmid \Gamma, \Theta, \Pi, q \nmid \Theta, \Pi, q \nmid \emptyset$}
\RightLabel{$(m_1^\ddagger)$}
\UnaryInfC{$q,\Sigma ; \Gamma, \unaryNotOpFb \ \Pi \Rightarrow^{\{ \mathbf{b}, \mathbf{n} \}}_{\logicFour} \Delta ; \Theta$}
\end{prooftree}
\end{minipage}
\begin{minipage}[t]{0.45\textwidth}
\begin{prooftree}
\def\fCenter{\mbox{\ $\Rightarrow^{\{ \mathbf{b}, \mathbf{n} \}}_{\logicFour}$\ }}
\AxiomC{$\Sigma, \Gamma \mid \Pi, \Gamma \mid \Pi, \Delta \mid \Sigma, \Delta$}
\RightLabel{$(m_2^\ddagger)$}
\UnaryInfC{$\Sigma ; \Gamma, \unaryNotOpFb \ \Pi \Rightarrow^{\{ \mathbf{b}, \mathbf{n} \}}_{\logicFour} \Delta ; \Theta$}
\end{prooftree}
\end{minipage}
\end{center}

\bigskip

\begin{figure}[t!]
\hrule
\medskip
\begin{center}

\begin{minipage}[t]{0.32\textwidth}
\begin{prooftree}
\AxiomC{$\Gamma\nmid\Delta\nmid\Pi\nmid\Omega,\varphi$}
\RightLabel{$(\neg : \mathbf{f})^\dashv$}
\UnaryInfC{$\Gamma,\neg \varphi \nmid\Delta\nmid\Pi\nmid\Omega$}
\end{prooftree}
\end{minipage}
\begin{minipage}[t]{0.32\textwidth}
\begin{prooftree}
\AxiomC{$\Gamma\nmid\Delta,\varphi\nmid\Pi\nmid\Omega$}
\RightLabel{$(\neg : \mathbf{n})^\dashv$}
\UnaryInfC{$\Gamma\nmid\Delta,\neg \varphi\nmid\Pi\nmid\Omega$}
\end{prooftree}
\end{minipage}
\begin{minipage}[t]{0.32\textwidth}
\begin{prooftree}
\AxiomC{$\Gamma\nmid\Delta\nmid\Pi,\varphi\nmid\Omega$}
\RightLabel{$(\neg : \mathbf{b})^\dashv$}
\UnaryInfC{$\Gamma\nmid\Delta\nmid\Pi,\neg \varphi\nmid\Omega$}
\end{prooftree}
\end{minipage}

\medskip

\begin{minipage}[t]{0.26\textwidth}
\begin{prooftree}
\AxiomC{$\Gamma,\varphi\nmid\Delta\nmid\Pi\nmid\Omega$}
\RightLabel{$(\neg : \mathbf{t})^\dashv$}
\UnaryInfC{$\Gamma\nmid\Delta\mid\Pi\nmid\Omega,\neg \varphi$}
\end{prooftree}
\end{minipage}
\begin{minipage}[t]{0.35\textwidth}
\begin{prooftree}
\AxiomC{$\Gamma, \varphi, \psi \nmid \Delta, \varphi, \psi \nmid \Pi \nmid \Omega $}
\RightLabel{$(\land : \mathbf{f}^1)^\dashv$}
\UnaryInfC{$\Gamma, \varphi \land \psi \nmid \Delta \nmid \Pi \nmid \Omega $}
\end{prooftree}
\end{minipage}
\begin{minipage}[t]{0.35\textwidth}
\begin{prooftree}
\AxiomC{$\Gamma, \varphi, \psi \nmid \Delta \nmid \Pi, \varphi, \psi \nmid \Omega $}
\RightLabel{$(\land : \mathbf{f}^2)^\dashv$}
\UnaryInfC{$\Gamma, \varphi \land \psi \nmid \Delta \nmid \Pi \nmid \Omega $}
\end{prooftree}
\end{minipage}

\medskip

\begin{minipage}[t]{0.32\textwidth}
\begin{prooftree}
\AxiomC{$\Gamma \nmid \Delta, \varphi, \psi \nmid \Pi \nmid \Omega $}
\RightLabel{$(\land : \mathbf{n}^1)^\dashv$}
\UnaryInfC{$\Gamma \nmid \Delta, \varphi \land \psi \nmid \Pi \nmid \Omega $}
\end{prooftree}
\end{minipage}
\begin{minipage}[t]{0.32\textwidth}
\begin{prooftree}
\AxiomC{$\Gamma \nmid \Delta, \varphi \nmid \Pi \nmid \Omega, \varphi $}
\RightLabel{$(\land : \mathbf{n}^2)^\dashv$}
\UnaryInfC{$\Gamma \nmid \Delta, \varphi \land \psi \nmid \Pi \nmid \Omega $}
\end{prooftree}
\end{minipage}
\begin{minipage}[t]{0.32\textwidth}
\begin{prooftree}
\AxiomC{$\Gamma \nmid \Delta, \psi \nmid \Pi \nmid \Omega, \psi $}
\RightLabel{$(\land : \mathbf{n}^3)^\dashv$}
\UnaryInfC{$\Gamma \nmid \Delta, \varphi \land \psi \nmid \Pi \nmid \Omega $}
\end{prooftree}
\end{minipage}

\medskip

\begin{minipage}[t]{0.32\textwidth}
\begin{prooftree}
\AxiomC{$\Gamma \nmid \Delta \nmid \Pi, \varphi, \psi \nmid \Omega $}
\RightLabel{$(\land : \mathbf{b}^1)^\vdash$}
\UnaryInfC{$\Gamma \nmid \Delta \nmid \Pi, \varphi \land \psi \nmid \Omega $}
\end{prooftree}
\end{minipage}
\begin{minipage}[t]{0.32\textwidth}
\begin{prooftree}
\AxiomC{$\Gamma \nmid \Delta \nmid \Pi, \varphi \nmid \Omega, \varphi $}
\RightLabel{$(\land : \mathbf{b}^2)^\dashv$}
\UnaryInfC{$\Gamma \nmid \Delta \nmid \Pi, \varphi \land \psi \nmid \Omega $}
\end{prooftree}
\end{minipage}
\begin{minipage}[t]{0.32\textwidth}
\begin{prooftree}
\AxiomC{$\Gamma \nmid \Delta \nmid \Pi, \psi \nmid \Omega, \psi $}
\RightLabel{$(\land : \mathbf{b}^3)^\dashv$}
\UnaryInfC{$\Gamma \nmid \Delta \nmid \Pi, \varphi \land \psi \nmid \Omega $}
\end{prooftree}
\end{minipage}

\medskip

\begin{minipage}[t]{0.32\textwidth}
\begin{prooftree}
\AxiomC{$\Gamma \nmid \Delta \nmid \Pi \nmid \Omega, \varphi $}
\RightLabel{$(\land : \mathbf{t}^1)^\dashv$}
\UnaryInfC{$\Gamma \nmid \Delta \nmid \Pi \nmid \Omega, \varphi \land \psi $}
\end{prooftree}
\end{minipage}
\begin{minipage}[t]{0.32\textwidth}
\begin{prooftree}
\AxiomC{$\Gamma \nmid \Delta \nmid \Pi \nmid \Omega, \psi $}
\RightLabel{$(\land : \mathbf{t}^2)^\dashv$}
\UnaryInfC{$\Gamma \nmid \Delta \nmid \Pi \nmid \Omega, \varphi \land \psi $}
\end{prooftree}
\end{minipage}
\begin{minipage}[t]{0.32\textwidth}
\begin{prooftree}
\AxiomC{$\Gamma, \varphi \nmid \Delta, \varphi \nmid \Pi\nmid\Omega, \psi $}
\RightLabel{$(\supset : \mathbf{t})^\dashv$}
\UnaryInfC{$\Gamma\nmid \Delta \nmid \Pi \nmid\Omega, \varphi \supset \psi  $}
\end{prooftree}
\end{minipage}

\medskip

\begin{minipage}[t]{0.32\textwidth}
\begin{prooftree}
\AxiomC{$\Gamma \nmid \Delta \nmid \Pi, \varphi \nmid \Omega, \varphi$}
\RightLabel{$(\supset : \mathbf{f}^1)^\dashv$}
\UnaryInfC{$\Gamma, \varphi \supset \psi \nmid \Delta \nmid \Pi \nmid\Omega $}
\end{prooftree}
\end{minipage}
\begin{minipage}[t]{0.32\textwidth}
\begin{prooftree}
\AxiomC{$\Gamma, \psi \nmid \Delta \nmid \Pi \nmid\Omega$}
\RightLabel{$(\supset : \mathbf{f}^2)^\dashv$}
\UnaryInfC{$\Gamma, \varphi \supset \psi \nmid \Delta \nmid \Pi \nmid\Omega $}
\end{prooftree}
\end{minipage}
\begin{minipage}[t]{0.32\textwidth}
\begin{prooftree}
\AxiomC{$\Gamma \nmid \Delta \nmid \Pi, \varphi \nmid \Omega, \varphi$}
\RightLabel{$(\supset : \mathbf{n}^1)^\dashv$}
\UnaryInfC{$\Gamma\nmid \Delta, \varphi \supset \psi  \nmid \Pi \nmid\Omega $}
\end{prooftree}
\end{minipage}

\medskip

\begin{minipage}[t]{0.32\textwidth}
\begin{prooftree}
\AxiomC{$\Gamma \nmid \Delta, \psi \nmid \Pi \nmid\Omega$}
\RightLabel{$(\supset : \mathbf{n}^2)^\dashv$}
\UnaryInfC{$\Gamma\nmid \Delta, \varphi \supset \psi  \nmid \Pi \nmid\Omega $}
\end{prooftree}
\end{minipage}
\begin{minipage}[t]{0.32\textwidth}
\begin{prooftree}
\AxiomC{$\Gamma \nmid \Delta \nmid \Pi, \varphi \nmid \Omega, \varphi$}
\RightLabel{$(\supset : \mathbf{b}^1)^\dashv$}
\UnaryInfC{$\Gamma\nmid \Delta \nmid \Pi , \varphi \supset \psi \nmid\Omega $}
\end{prooftree}
\end{minipage}
\begin{minipage}[t]{0.32\textwidth}
\begin{prooftree}
\AxiomC{$\Gamma \nmid \Delta \nmid \Pi, \psi \nmid \Omega$}
\RightLabel{$(\supset : \mathbf{b}^2)^\dashv$}
\UnaryInfC{$\Gamma\nmid \Delta \nmid \Pi , \varphi \supset \psi \nmid\Omega $}
\end{prooftree}
\end{minipage}

\end{center}

\hrule
\caption{Rules of the anti-sequent calculus $\mathsf{R}_\logicFour$.}\label{fig:fCalcNegative}
\end{figure}

In view of our general construction in Section~\ref{sec:genMethod}, we have the following result:

\begin{theorem}
	Let $\Gamma$ and $\Delta$ be $\mathcal{F}$-theories. Then,
	\begin{enumerate}[{\em (}i{\em )}]
	\item $\Gamma \modelsFourMa \Delta$ iff $\emptyset ; \Gamma \Rightarrow^{\{ \mathbf{b} \}}_{\logicFour} \Delta ; \mathit{Var}(\Gamma \cup \Delta)$ is provable in $\mathsf{ME}^1_\logicFour$, and
	
	\item $\Gamma \modelsFourMb \Delta$ iff $\emptyset ; \Gamma \Rightarrow^{\{ \mathbf{b}, \mathbf{n} \}}_{\logicFour} \Delta ; \mathit{Var}(\Gamma \cup \Delta)$) is provable in $\mathsf{ME}^2_\logicFour$,
	\end{enumerate}
where $\mathit{Var}(\Gamma \cup \Delta)$ is the set of propositional constants appearing in $\Gamma$ or $\Delta$.
\end{theorem}

To conclude our discussion, we give an example illustrating a proof 
in $\mathsf{ME}^1_\logicFour$.
\begin{example}
Consider again the theory $$\Gamma = \{ p, \neg (p \land \neg q) \}$$ from Example~\ref{ex:smp}. As $\Gamma \modelsFourMa q$ holds, the sequent 
$$\emptyset ; p, \neg (p \land \neg q) \Rightarrow^{\{ \mathbf{b} \}}_{\logicFour} q ; p,q$$
is provable in $\mathsf{ME}^1_\logicFour$. A proof, $\beta$, of this sequent is given below, using the subproof $\alpha$:

\begin{itemize} 
\item Proof $\alpha$:
\vspace{-1.5ex}
\begin{prooftree}
\def\defaultHypSeparation{\hskip 1mm}
\AxiomC{$p, q \mid p,p,q \mid p,q \mid q$}
\RightLabel{$(\neg : \mathbf{n})^\vdash$}
\UnaryInfC{$p, q \mid p,p, \neg q \mid p,q \mid q$}
\AxiomC{$p, q \mid p,p \mid p,q \mid q, p$}
\AxiomC{$p, q \mid p, p, q \mid p,q \mid q, \neg q $}
\RightLabel{$(\neg : \mathbf{n})^\vdash$}
\UnaryInfC{$p, q \mid p, p, \neg q \mid p,q \mid q, \neg q $}
\RightLabel{$(\land : \mathbf{n})^\vdash$}
\TrinaryInfC{$p, q \mid p, p \land \neg q \mid p,q \mid q $}	
 \RightLabel{ $(\neg : \mathbf{t})^\vdash$}
 \UnaryInfC{$p \mid p, p \land \neg q \mid p,q \mid q, \neg q$}
\end{prooftree}

\item Proof $\beta$:
\vspace{-1.5ex}
 \begin{prooftree}
 \def\defaultHypSeparation{\hskip 2mm}
 \def\labelSpacing{1pt}
 \AxiomC{$p,q \nmid p, q \nmid q, p \nmid \emptyset$}
 \RightLabel{ $(\neg : \mathbf{n})^\dashv$}
 \UnaryInfC{$p,q \nmid p, \neg q \nmid q, p \nmid \emptyset$}
 \RightLabel{ $(w : \mathbf{n})^\dashv$}
 \UnaryInfC{$p,q \nmid p, p, \neg q \nmid q, p \nmid \emptyset$}
 \RightLabel{$(\land : \mathbf{n})^\dashv$}
 \UnaryInfC{$p,q \nmid p, p \land \neg q \nmid q, p \nmid \emptyset$}
 \RightLabel{ $(\neg : \mathbf{t})^\dashv$}
 \UnaryInfC{$p \nmid p, p \land \neg q \nmid q, p \nmid \neg q$}
 \RightLabel{ $(\land : \mathbf{t})^\dashv$}
 \UnaryInfC{$p \nmid p, p \land \neg q \nmid q, p \nmid p \land \neg q$}
 \RightLabel{ $(\neg : \mathbf{n})^\dashv$}
 \UnaryInfC{$p \nmid p, \neg (p \land \neg q) \nmid q, p \nmid p \land \neg q$}
 \RightLabel{ $(\neg : \mathbf{f})^\dashv$}
 \UnaryInfC{$p, \neg (p \land \neg q) \nmid p, \neg (p \land \neg q) \nmid q, p \nmid \emptyset$}
 \RightLabel{ $(m_1)$}
 \UnaryInfC{$p ; p, \neg (p \land \neg q) \Rightarrow^{\{ \mathbf{b} \}}_{\logicFour} q ; q$}
 \AxiomC{$p \mid p, p \land \neg q \mid p,q \mid q, p$}
 \AxiomC{$\alpha$}
 \RightLabel{ $(\land : \mathbf{t})^\vdash$}
 \BinaryInfC{$p \mid p, p \land \neg q \mid p,q \mid q, p \land \neg q$}
 \RightLabel{ $(\neg : \mathbf{n})^\vdash$}
 \UnaryInfC{$p \mid p, \neg (p \land \neg q) \mid p,q \mid q , p \land \neg q$}
 \RightLabel{ $(\neg : \mathbf{t})^\vdash$}
 \UnaryInfC{$p, \neg (p \land \neg q) \mid p, \neg (p \land \neg q) \mid p,q \mid q$}
 \RightLabel{ $(m_2)$}
 \UnaryInfC{$\emptyset ; p, \neg (p \land \neg q), \unaryNotOpFa p \Rightarrow^{\{ \mathbf{b} \}}_{\logicFour} q ; q$}
 \RightLabel{ $(m_3)$}
 \BinaryInfC{$\emptyset ; p, \neg (p \land \neg q) \Rightarrow^{\{ \mathbf{b} \}}_{\logicFour} q ; p,q$}
 \end{prooftree}

\end{itemize}
  
\end{example}

\section{Conclusion}\label{sec:concl}

In this paper, we introduced a general method for obtaining sound and complete sequent-type calculi for a whole class of nonmonotonic minimal-entailment relations in the style of the proof-theoretical approach 
due to Bonatti and Olivetti~\cite{bonatti2002}.
We obtained particular calculi for well-known paraconsistent logics as special instances of our general method. 

Concerning future work, it would be interesting to develop a similar proof-theoretical approach for more general lattice-based entailment relations than the one studied here, like those discussed by Ginsberg~\cite{ginsberg1988}. Moreover, generalisations
to the predicate-logic case would be a worthwhile endeavour too.


\end{document}